\documentclass[10pt,a4paper]{amsart}
\usepackage[margin=20mm]{geometry}
\usepackage[numbers]{natbib}
\usepackage{hyperref}

\usepackage{amssymb,amsmath}
\usepackage{amsthm}
\usepackage{bbm}
\usepackage{stmaryrd}
\usepackage[usenames,dvipsnames]{xcolor}
\usepackage{color}

\input xy
\xyoption{all} 

\numberwithin{equation}{section}

\newtheorem{theorem}{Theorem}[section]
\newtheorem{corollary}[theorem]{Corollary}
\newtheorem{lemma}[theorem]{Lemma}
\newtheorem{proposition}[theorem]{Proposition}

\theoremstyle{definition}
\newtheorem{definition}[theorem]{Definition}
\newtheorem{remark}[theorem]{Remark}
\newtheorem{example}[theorem]{Example}

\newenvironment{warning}[1][Warning.]{\begin{trivlist}
\item[\hskip \labelsep {\bfseries #1}]}{\end{trivlist}}
\newenvironment{statement}[1][Statement.]{\begin{trivlist}
\item[\hskip \labelsep {\bfseries #1}]\it  }{ \end{trivlist}}

\newcommand{\Id}{\mathbbmss{1}}

\DeclareMathOperator{\Vect}{Vect}

\DeclareMathOperator{\Span}{Span}

\font\black=cmbx10 \font\sblack=cmbx7 \font\ssblack=cmbx5 \font\blackital=cmmib10  \skewchar\blackital='177
\font\sblackital=cmmib7 \skewchar\sblackital='177 \font\ssblackital=cmmib5 \skewchar\ssblackital='177
\font\sanss=cmss10 \font\ssanss=cmss8 
\font\sssanss=cmss8 scaled 600 \font\blackboard=msbm10 \font\sblackboard=msbm7 \font\ssblackboard=msbm5
\font\caligr=eusm10 \font\scaligr=eusm7 \font\sscaligr=eusm5  \font\fraktur=eufm10
\font\sfraktur=eufm7 \font\ssfraktur=eufm5 
\font\bsymb=cmsy10 scaled\magstep2
\def\all#1{\setbox0=\hbox{\lower1.5pt\hbox{\bsymb
       \char"38}}\setbox1=\hbox{$_{#1}$} \box0\lower2pt\box1\;}
\def\exi#1{\setbox0=\hbox{\lower1.5pt\hbox{\bsymb \char"39}}
       \setbox1=\hbox{$_{#1}$} \box0\lower2pt\box1\;}

\def\tx#1{{\fam0\relax#1}}

\newfam\bifam
\textfont\bifam=\blackital \scriptfont\bifam=\sblackital \scriptscriptfont\bifam=\ssblackital

\newfam\blfam
\textfont\blfam=\black \scriptfont\blfam=\sblack \scriptscriptfont\blfam=\ssblack

\newfam\bbfam
\textfont\bbfam=\blackboard \scriptfont\bbfam=\sblackboard \scriptscriptfont\bbfam=\ssblackboard

\newfam\ssfam
\textfont\ssfam=\sanss \scriptfont\ssfam=\ssanss \scriptscriptfont\ssfam=\sssanss
\def\sss#1{{\fam\ssfam\relax#1}}

\newfam\clfam
\textfont\clfam=\caligr \scriptfont\clfam=\scaligr \scriptscriptfont\clfam=\sscaligr

\newfam\frfam
\textfont\frfam=\fraktur \scriptfont\frfam=\sfraktur \scriptscriptfont\frfam=\ssfraktur

\def\hpb#1{\setbox0=\hbox{${#1}$}
    \copy0 \kern-\wd0 \kern.2pt \box0}
\def\vpb#1{\setbox0=\hbox{${#1}$}
    \copy0 \kern-\wd0 \raise.08pt \box0}

\def\pmb#1{\setbox0\hbox{${#1}$} \copy0 \kern-\wd0 \kern.2pt \box0}
\def\pmbb#1{\setbox0\hbox{${#1}$} \copy0 \kern-\wd0
      \kern.2pt \copy0 \kern-\wd0 \kern.2pt \box0}
\def\pmbbb#1{\setbox0\hbox{${#1}$} \copy0 \kern-\wd0
      \kern.2pt \copy0 \kern-\wd0 \kern.2pt
    \copy0 \kern-\wd0 \kern.2pt \box0}
\def\pmxb#1{\setbox0\hbox{${#1}$} \copy0 \kern-\wd0
      \kern.2pt \copy0 \kern-\wd0 \kern.2pt
      \copy0 \kern-\wd0 \kern.2pt \copy0 \kern-\wd0 \kern.2pt \box0}
\def\pmxbb#1{\setbox0\hbox{${#1}$} \copy0 \kern-\wd0 \kern.2pt
      \copy0 \kern-\wd0 \kern.2pt
      \copy0 \kern-\wd0 \kern.2pt \copy0 \kern-\wd0 \kern.2pt
      \copy0 \kern-\wd0 \kern.2pt \box0}

\mathchardef\za="710B  
\mathchardef\zb="710C  
\mathchardef\zg="710D  
\mathchardef\zd="710E  
\mathchardef\zve="710F 
\mathchardef\zz="7110  
\mathchardef\zh="7111  
\mathchardef\zvy="7112 
\mathchardef\zi="7113  
\mathchardef\zk="7114  
\mathchardef\zl="7115  
\mathchardef\zm="7116  
\mathchardef\zn="7117  
\mathchardef\zx="7118  
\mathchardef\zp="7119  
\mathchardef\zr="711A  
\mathchardef\zs="711B  
\mathchardef\zt="711C  
\mathchardef\zu="711D  
\mathchardef\zvf="711E 
\mathchardef\zq="711F  
\mathchardef\zc="7120  
\mathchardef\zw="7121  
\mathchardef\ze="7122  
\mathchardef\zy="7123  
\mathchardef\zf="7124  
\mathchardef\zvr="7125 
\mathchardef\zvs="7126 
\mathchardef\zf="7127  
\mathchardef\zG="7000  
\mathchardef\zD="7001  
\mathchardef\zY="7002  
\mathchardef\zL="7003  
\mathchardef\zX="7004  
\mathchardef\zP="7005  
\mathchardef\zS="7006  
\mathchardef\zU="7007  
\mathchardef\zF="7008  
\mathchardef\zW="700A  
\mathchardef\zC="7009  

\newcommand{\be}{\begin{equation}}
\newcommand{\ee}{\end{equation}}

\newcommand{\bea}{\begin{eqnarray}}
\newcommand{\eea}{\end{eqnarray}}
\def\*{{\textstyle *}}
\newcommand{\R}{{\mathbb R}}

\newcommand{\C}{{\mathbb C}}
\newcommand{\Z}{{\mathbb Z}}

\newcommand{\s}{{\textstyle *}}

\def\Vect{\sss{Vect}}

\def\sT{{\sss T}}

\def\xi{\tx{i}}

\def\s*{{\scriptstyle *}}

\def\cO{\mathcal{O}}

\newcommand{\beas}{\begin{eqnarray*}}
\newcommand{\eeas}{\end{eqnarray*}}

\def\half{\frac{1}{2}}

\title{Odd Connections on Supermanifolds: Existence and Relation with Affine Connections }

   \author{Andrew James Bruce$^\dag$ \& Janusz Grabowski$^\ddag$}
 \address{ $^\dag$Mathematics Research Unit, University of Luxembourg,  Esch-sur-Alzette, Luxembourg \\
   \newline  $^\ddag$ Institute of Mathematics, Polish Academy of Sciences, Poland}
   \email{andrewjamesbruce@googlemail.com,~jagrab@impan.pl}

\date{\today}
 
\begin{document}

\begin{abstract}
The notion of an \emph{odd quasi-connection} on a supermanifold, which is loosely an affine connection that carries non-zero Grassmann parity, is  examined. Their torsion and curvature are defined, however, in general, they are not tensors.  A special  class of such generalised connections, referred to as \emph{odd connections} in this paper, have torsion and curvature  tensors.  Part of the structure is an odd involution of the tangent bundle of the supermanifold and this puts drastic restrictions on the supermanifolds that admit odd connections. In particular, they must have equal number of even and odd dimensions. Amongst other results, we show that an odd connection is defined, up to an odd tensor field of type $(1,2)$, by an affine connection and an odd endomorphism of the tangent bundle. Thus, the theory of odd connections and  affine connections are not completely separate  theories.  As an example relevant to physics, it is shown that  $\mathcal{N}=1$ super-Minkowski spacetime admits a natural odd connection.
\par
\smallskip\noindent
{\bf Keywords:}
Supermanifolds;~Affine Connections;~ Quasi-Connections;~Lie Supergroups.\par
\smallskip\noindent
{\bf MSC 2010:}~16W50;~17B66;~53B05;~58A50. 	
\end{abstract}

 \maketitle

\setcounter{tocdepth}{2}
 \tableofcontents

\section{Introduction}
 It hardly needs to be mentioned, but the notion of a connection in its various guises is of central importance in differential geometry and geometric approaches to physics.   A prominent example of the r\^{o}le of connections in modern mathematics is the construction of characteristic classes of principle bundles via  Chern--Weil  theory. In physics, connections are related to gauge fields and are vital in geometric approaches to relativistic mechanics, general  relativity and  other geometric approaches to gravity such as metric-affine gravity,  Fedosov's  deformation quantisation, adiabatic evolution via the Berry phase, and so on. For an overview of connections in classical and quantum field theory the reader may consult \cite{Mangiarotti:2000}. Over the years there have been many generalisations of a connection on a manifold given in the literature, including the generalisation to Lie algebroids, Courant algebroids (see \cite{Gualtieri:2010}) and connections adapted to  non-negatively graded manifolds (see \cite{Bruce:2019}), to name a few.  In the noncommutative setting, we have, for example, linear connections on bimodules over almost commutative algebras (see \cite{Ciupala:2003}).  The situation with connections in general with noncommutative geometry  is more subtle and depends on the approach taken. A brief discussion of this and the notion of $q$-deformed Levi-Civita connections can be found in the preprint \cite{Arnlind:2020}.  \par
 Supersymmetry has been an attractive subject to theoretical physicists since its inception in the early 1970s. This is, in part, because supersymmetric field theories often have desirable mathematical properties, such as milder divergences and  in very special cases the theories can be finite. Supersymmetry also  removes the tachyon from the spectrum of string theories and naturally leads to a theory of gravity when promoted to a local gauge theory. Alongside the developments in physics, supergeometry, i.e., $\Z_2$-graded geometry, has become a respectable branch of mathematics with deep and surprising links with not just physics, but also classical differential geometry, homological and homotopical algebra, to name a few. A prime example of the aforementioned mathematical insight is Voronov's approach to Drinfeld doubles for Lie bialgebroids (see \cite{Voronov:2002}). The `operational' use of supergeometry goes back to the early days of supersymmetry with the superspace methods of Salam \& Strathdee \cite{Salam:1974}. Superspace methods provide an elegant way of constructing supersymmetric actions and are routinely used today. To give mathematical rigour to the notion of a superspace,  Berezin \& Leties \cite{Berezin:1976} defined a supermanifold in terms of algebraic geometry, specifically using locally superringed spaces. Much of the fundamental work on supergeometry was carried out between 1965 and 1975 by Berezin and his collaborators.  However, we must stress that the theory of supermanifolds is well-motivated independently of supersymmetry. For instance, any geometric formulation of pseudo-classical theories with fermions, ghost fields, antifields etc requires supergeometry.  \par
The notion of a connection, particularly Koszul's algebraic notion, generalises to the category of supermanifolds rather directly, in essence, one needs to insert the correct plus and minus signs into the classical definitions. Connections on supermanifolds appear in the context of Fedosov supermanifolds \cite{Geyer:2004}, the BV-formalism \cite{Batalin:2008} and natural quantisation of supermanifolds \cite{Leuther:2011}, for example. It is well-known that the fundamental theorem of Riemannian geometry generalised to supermanifolds equipped with either an even or odd Riemannian metric, see for example \cite{Monterde:1996}. As a historical remark, one of the earliest papers on supergravity is rooted in Riemannian supergeometry, though at the time the theory of supermanifolds was in its infancy (see \cite{Aronwitt:1976}).   Importantly from the perspective of this paper, an affine connection on a supermanifold is an even object. That is, the parity of the connection itself is zero.  In this paper, we address the notion of an affine connection on a supermanifold that is odd, i.e., carries Grassmann parity one.  Such a concept has not appeared in the literature before. \par
Our approach to odd connections on supermanifolds is very similar to the notion of a quasi-connection as first defined by Y-C. Wong \cite{Wong:1962} in 1962, which is related to the notion of a connection  on a  Lie algebroid  as first introduced by  Mackenzie in the transitive case \cite{Mackenzie:1987}, and  a connection over a vector bundle map as defined by Cantrijn \&  Langerock \cite{Cantrijn:2003}. However the presence of a $\Z_2$-grading and, in particular,  the fact that we want odd objects means that we cannot directly translate all of Wong's constructions to our setting. Similarly, our notion of an odd connection is not simply a specialisation of a Lie algebroid connection. For a review of quasi-connections and further references, the reader may consult Etayo \cite{Etayo:1993}. We remark that the notions we put forward are not to be confused with Quillen's notion of a superconnection (see \cite{Quillen:1985}).   \par
The motivation for this work stems from the philosophy that alongside the Grassmann even generalisations of classical notions in differential geometry,    Grassmann odd analogues can also be found. Although odd structures have no classical counterpart, they should still be treated on equal footing as even structures.  As prime examples, we have even and odd Riemannian structures, symplectic/Poisson structures and contact/Jacobi structures. Most of these odd structures have found some application in physics, in particular, odd symplectic/Poisson structures are central to the BV-formalism and its generalisations, see for example \cite{Kazinski:2005, Khudaverdian:1991,Khudaverdian:2004,Khudaverdian:2002,Lyakhovich:2004,Schwarz:1993}. The notable exception here are odd Riemannian structures, which so far have not found an application in physics.  Furthermore,  odd counterparts of superconformal transformations that twist the parity of the standard basis of the module of vector fields on $\C^{1|1}$, known as  $\textrm{TPt}$ transformations, were first proposed in \cite{Duplij:1991, Duplij:1997}, and, in the same papers, led to an odd generalisation of super Riemann surfaces.   These  generalisations spawned an odd analogue of a complex structure and  were motivated by developments in two-dimensional superconformal field theory and its applications to string theory. With these observations in mind, the natural question of the notion of a connection on a supermanifold that carries non-zero Grassmann parity arises.  Alongside this, if a good concept of an odd connection exists, then do any of the supermanifolds of interest in physics admit such things?   Is there any relation with supersymmetry as formulated in superspace?
\smallskip

\noindent \textbf{Main Results:}  Loosely, an odd quasi-connection consists of an odd linear map $\nabla : \Vect(M) \times \Vect(M) \rightarrow \Vect(M)$ and an odd endomorphism $ \rho : \Vect(M)  \rightarrow  \Vect(M) $, together with a compatibility condition between the two,  which is just a graded Leibniz rule, see Definition \ref{Def:OddQuaCon} for details. There are natural generalisations of the torsion and curvature for odd quasi-connections, see Definition \ref{Def:Torsion} and Definition \ref{Def:RieCurv}. In general, these are not tensors. Amongst other results, we have the following.
\begin{enumerate}
 \setlength\itemsep{1em}
\item If $\rho$ is an odd involution, then the torsion and curvature are tensors, see Theorem \ref{Thm:TorCurTensors}. Such odd quasi-connections we refer to as \emph{odd connections}. An important result here is that odd connections and affine connections  are not entirely separate notions. In particular, an odd connection is defined by an affine connection and an odd involution, up to an odd tensor field. Conversely, an odd connection canonically defines an affine connection. Essentially, if $\bar{\nabla}$ is an affine connection, then
$$\nabla_X := \bar{\nabla}_{\rho(X)}$$
 is an odd connection and, as $\rho^2 = \Id$, the converse statement is also true.  See Proposition \ref{Prop:CanGen} and Proposition \ref{Prop:ConGenII} for details.
\item The   curvature and torsion of an odd connection satisfy a generalised version of the algebraic Bianchi identity, see Theorem \ref{Thm:Bianchi1}. We view this identity as a compatibility between the  curvature and torsion.
\item We prove that $n|n$-dimensional Lie supergroups always admit an odd connection, see Theorem \ref{Thm:ExodConLieGrp}.  More generally than this,  $n|n$-dimensional parallelisable supermanifolds  always admit odd connections, see Theorem \ref{Thm:OddConPara}.
\item We show that $d=4$, $\mathcal{N} =1$ super-Minkowski spacetime comes equipped with a natural odd connection that we refer to as the {SUSY odd connection} (see Definition \ref{Def:SUSYOddCon}). Moreover, this odd connection is flat but has non-zero torsion, see Proposition \ref{Prop:SMinkFlat} and Proposition  \ref{Prop:SMinkTorsion}.
\item The example of super-Minkowski space-time leads to the notion of an odd Weitzenb\"{o}ck connection (see Definition \ref{Def:OddWeitzCon}) on a  $n|n$-dimensional parallelisable supermanifold. We show that such connections only depend on the existence of an odd involution and as such $n|n$-dimensional parallelisable supermanifold always admit odd Weitzenb\"{o}ck connection, see Proposition \ref{Prop:ExistWeitz}. Furthermore, it is shown that an odd Weitzenb\"{o}ck connection is compatible with an odd Riemannian metric, see  Proposition \ref{Prop:OddWeitComOddMet}.
\end{enumerate}
 In short, we have a reasonable theory of odd connections on a supermanifold,  even if the theory is not completely distinct from the theory of affine connections. Moreover, some of the supermanifolds of physical interest can be equipped with such structures.
\medskip

\noindent \textbf{Notation and preliminary concepts:} We will assume that the reader has a grasp of the basic theory of supermanifolds. For overviews of the general theory the reader may consult, for example, \cite{Carmeli:2011,Manin:1997,Varadrajan:2004}. We understand a \emph{supermanifold} $M := (|M|, \:  \cO_{M})$ of dimension $n|m$  to be a supermanifold as defined by Berezin \& Leites  \cite{Berezin:1976,Leites:1980}, i.e., as a locally superringed space that is locally isomorphic to $\mathbb{R}^{n|m} := \big (\R^{n}, C^{\infty}(\R^{n})\otimes \Lambda(\zx^{1}, \cdots \zx^{m}) \big)$. Here, $\Lambda(\zx^{1}, \cdots \zx^{m})$ is the Grassmann algebra (over $\R$) with $m$ generators.   Associated with any supermanifold is the sheaf morphism $\epsilon_{-} : \cO_M(-) \rightarrow C^\infty_{|M|}(-)$, which means that the underlying topological space $|M|$ is, in fact, a smooth manifold. This manifold we refer to as the \emph{reduced manifold}. Morphisms of supermanifolds are morphisms as  superringed spaces. That is, a morphism $\phi : M \rightarrow N$ consists of a pair $ \phi = (|\phi|, \phi^*  )$, where $|\phi| : |M| \rightarrow |N|$ is a continuous map (in fact, smooth) and  $\phi^*$  is a family of superring morphisms $\phi^*_{|V|} : \cO_N(|V|) \rightarrow \cO_M\big( |\phi|^{-1}(|V|)\big)$, for every open $|V| \subset |N|$, that respect the restriction maps. Given any point on $|M|$ we can always find a `small enough' open neighbourhood $|U|\subseteq |M|$ such that we can  employ local coordinates $x^{a} := (x^{\mu} , \zx^{i})$ on $M$.
It is well-known that morphisms between supermanifolds are completely described by their local coordinate expressions. In particular, changes of local coordinates we will write, using the standard abuses of notation, as $x^{a'} = x^{a'}(x)$.
The (global) sections of the structure sheaf we will refer to as \emph{functions}. The  supercommutative algebra of functions we will denote as $C^{\infty}(M)$.     The Grassmann parity of an object $A$ will be denoted by `tilde', i.e., $\widetilde{A} \in \Z_{2}$.  By `even' or `odd' we will be referring to the  Grassmann parity of an object. Note that as we are dealing with real supermanifolds in the locally ringed space approach, partitions of unity and bump functions always exist (see \cite[Lemma 3.1.7 and Corollary 3.1.8]{Leites:1980}).\par
The \emph{tangent sheaf} $\mathcal{T}M$ of a supermanifold $M$ is  defined as the sheaf of derivations of sections of the structure sheaf. Naturally, this is  a sheaf of locally free $\cO_{M}$-modules of rank $n|m$. Global sections of the tangent sheaf are  referred to as \emph{vector fields}. We denote the $\cO_{M}(|M|)$-module of vector fields as $\Vect(M)$. The total space of the tangent sheaf $\sT M$ we will refer to  as the \emph{tangent bundle}.

\section{Odd Quasi-Connections, their Torsion and Curvature}
\subsection{Odd Quasi-Connections}
Modifying the definition of a quasi-connection as first given by Wong \cite{Wong:1962}, to the setting of supermanifolds and odd maps of modules, we propose the following definition.
\begin{definition}\label{Def:OddQuaCon}
An \emph{odd quasi-connection} on a supermanifold $M$ is a pair $(\nabla, \, \rho)$, where
$$\nabla : \Vect(M) \times \Vect(M) \longrightarrow \Vect(M)$$
is a bi-linear map, written as $(X,Y) \mapsto \nabla_X Y$, and
$$ \rho : \Vect(M)  \longrightarrow  \Vect(M) $$
is an odd $C^\infty(M)$-module endomorphism, that satisfy the following for all (homogeneous) $X$ and $Y \in \Vect(M)$ and $f \in C^\infty(M)$:
\begin{enumerate}
\item $\widetilde{\nabla_X Y} = \widetilde{X} + \widetilde{Y} + 1$,
\item $\nabla_{f\,X}Y = (-1)^{\widetilde{f}} f \, \nabla_X Y$,
\item $\nabla_X f\,Y =  \rho(X)f \, Y + (-1)^{(\widetilde{X} +1) \widetilde{f}} f \, \nabla_X Y$.
\end{enumerate}
\end{definition}
\begin{remark}
The reader should also note the similarity and differences with a Lie algebroid connection where the anchor map plays the analogue r\^{o}le to the odd endomorphism in the above definition. Also, note that at this stage there are no further conditions on the odd endomorphism.
\end{remark}
\begin{remark}
The notion of an odd quasi-connection could be reformulated as an \emph{even} map $\nabla : \Vect(M) \rightarrow  \Omega^1(M) \otimes \Vect(M)$ that satisfies the Leibniz rule $\nabla(f\,X) = \rho(f)\otimes X + f\; \nabla X$, where $\rho : C^\infty(M) \rightarrow \Omega^1(M)$ is an even map, it serves as a replacement to the de Rham differential in the standard setting of connections. We choose not to adopt this point of view in order to mimic the existing constructions related to quasi-connections following Wong \cite{Wong:1962}. Moreover, in physics one usually requires a connection understood as a covariant derivative.
\end{remark}
\begin{proposition}\label{Prop:AffineModule}
The set of all odd quasi-connections on a supermanifold $M$ is and affine space and a $C^\infty_0(M)$-module.
\end{proposition}
\begin{proof}
To show that we have the structure of an affine space, let $(\nabla, \, \rho)$ and $(\nabla', \, \rho')$ be odd quasi-connections on a supermanifold $M$. Then we claim that
$$(t\,\nabla + (1-t)\nabla', \, t \,\rho + (1-t)\rho' )$$
 is an odd quasi-connection for all $t \in \R$. It is easy to verify that the defining properties of an odd quasi-connection are satisfied. In particular, the parity is obvious and the other two properties follow from short computations. We leave details to the reader. \par
 Similarly, to show that we have  module we need to argue that
$$(f \, \nabla + \nabla' , \, f \,\rho + \rho')$$
 is an odd quasi-connection for an arbitrary $f \in C^\infty_0(M)$.  The function $f$ must be degree  zero in order to preserve the Grassmann parity. The remaining two properties follow from short computations. We again leave details to the reader.
\end{proof}
An important property of affine connections is that they are local operators, which implies that they have well-defined local expressions. The same is true of odd quasi-connections. This  is almost obvious in light of the fact  that odd quasi-connections are linear operators and satisfy a Leibniz rule.
\begin{proposition}
An odd quasi-connections $(\nabla, \, \rho)$ on a supermanifold $M$ is a local operator.
\end{proposition}
\begin{proof}
The proof follows in the same way as it does in the classical setting by using a bump function. Let $p \in |V|$ (open) and let  $|W| \subset |V| \subset |M|$ be a compact neighbourhood of $p$. We know that there exists a bump function $\gamma \in C^\infty_0(M)$ which restricts to $1$ on $|W|$ and whose support is included in $|V|$.  Hence, if $\gamma X$ vanishes on $|V|$, for some vector field $X$, then it also vanishes on $|M|\setminus \textnormal{supp}(\gamma)$.  \par
Then
$$0 = (\nabla_{\gamma X}Y)|_{|W|}=  (\gamma \,\nabla_{ X}Y)|_{|W|} = \gamma|_{|W|} \,(\nabla_{X} Y)|_{|W|} = (\nabla_{ X}Y)|_{|W|}. $$
Hence $(\nabla_{X}Y)|_{|V|} =0$ if $X|_{|V|} =0$. Similarly,
$$0 = (\nabla_ X \gamma Y)|_{|W|} =(\rho(X) \gamma \, Y)|_{|W|} + (\gamma \, \nabla_ X Y)|_{|W|} = (\rho(X) \gamma)|_{|W|} \, Y|_{|W|} + \gamma|_{|W|} \,( \nabla_ X Y)|_{|W|} =  (\nabla_{ X}Y)|_{|W|}.$$
Hence $(\nabla_{X}Y)|_{|V|} =0$ if $Y|_{|V|} =0$.
\end{proof}
An odd quasi-connection  has the following local form
\begin{equation}\label{Eqn:LocForOddQuaiCon}
\nabla_X Y = (-1)^{\widetilde{X} + \widetilde{a}}\, X^a\left (  \rho_a^{\,\,b}\frac{\partial Y^c}{\partial x^b} + (-1)^{(\widetilde{a} +1)(\widetilde{Y} + \widetilde{b})} \, Y^b \, \Gamma_{ba}^{\,\,\, \,c} \right) \frac{\partial}{\partial x^c},
\end{equation}
where $\widetilde{\rho_a^{\,\,b}} = \widetilde{a} + \widetilde{b} +1$ and $\widetilde{\Gamma_{ba}^{\,\,\, \,c}} = \widetilde{a} +\widetilde{b} +\widetilde{c} +1$.
\begin{proposition}\label{Prop:TransOddQuasiCon}
Under a  change of coordinates $x^{a'} = x^{a'}(x)$ the local structure functions of an odd quasi-connection transform as
\begin{align*}
(-1)^{\widetilde{a}'}\, \rho_{a'}^{\,\,b'} & = (-1)^{\widetilde{a}}\, \left(\frac{\partial x^a}{\partial x^{a'}} \right)\rho_a^{\,\,b}\left(\frac{\partial x^{b'}}{\partial x^b} \right) ,\\
 (-1)^{\widetilde{a}'}\,\Gamma_{b'a'}^{\,\,\,\,\, \,d'} &= (-1)^{ (\widetilde{a}+1)(\widetilde{b} + \widetilde{b}') + \widetilde{a}}\,  \left(\frac{\partial x^a}{\partial x^{a'}} \right)\left(\frac{\partial x^b}{\partial x^{b'}}  \right)\Gamma_{ba}^{\,\,\, \,c} \left(\frac{\partial x^{d'}}{\partial x^c} \right)\\
&+ (-1)^{\widetilde{a}}\, \left(\frac{\partial x^a}{\partial x^{a'}} \right)\rho_a^{\,\,c}\left(\frac{\partial x^{c'}}{\partial x^c} \right) \frac{\partial^2 x^d}{\partial x^{c'}\partial x^{b'}}\left(\frac{\partial x^{d'}}{\partial x^d} \right).
\end{align*}
\end{proposition}
\begin{proof}
The proof follows in more-or-less the same way as it does for affine connections on manifolds.  Directly from \eqref{Eqn:LocForOddQuaiCon}  and using the chain rule we have
\begin{align*}
\nabla_X Y & = (-1)^{\widetilde{X} + \widetilde{a}}\, X^a\left (  \rho_a^{\,\,b}\frac{\partial Y^c}{\partial x^b} + (-1)^{(\widetilde{a} +1)(\widetilde{Y} + \widetilde{b})} \, Y^b \, \Gamma_{ba}^{\,\,\, \,c} \right) \frac{\partial}{\partial x^c}\\
& = (-1)^{\widetilde{X} + \widetilde{a}}\, X^{a'}\frac{\partial x^{a}}{\partial x^{a'}}\left (  \rho_a^{\,\,b}  \frac{\partial x^{b'}}{\partial x^{b}}\frac{\partial}{\partial x^{b'}}\left (Y^{c'} \frac{\partial x^{c}}{\partial x^{c'}} \right)  + (-1)^{(\widetilde{a} +1)(\widetilde{Y} + \widetilde{b})} \, Y^{b'}\frac{\partial x^{b}}{\partial x^{b'}} \, \Gamma_{ba}^{\,\,\, \,c} \right)\frac{\partial x^{d'}}{\partial x^{c}} \frac{\partial}{\partial x^{d'}}\\
& = (-1)^{\widetilde{X} + \widetilde{a}}\, X^{a'}\left( \frac{\partial x^{a}}{\partial x^{a'}}\rho_a^{\,\,b}  \frac{\partial x^{b'}}{\partial x^{b}}\frac{\partial Y^{d'}}{\partial x^{b'}} + (-1)^{(\widetilde{Y} + \widetilde{b}')(\widetilde{a}' +1)}\, Y^{b'} \frac{\partial x^{a}}{\partial x^{a'}}\rho_a^{\,\,c}  \frac{\partial x^{c'}}{\partial x^{c}} \frac{\partial^2 x^b}{\partial x^{c'}\partial x^{b'}}\frac{\partial x^{d'}}{\partial x^{b}} \right.\\
& \left. +  \,(-1)^{(\widetilde{a}'+1)(\widetilde{Y} + \widetilde{b}')+ (\widetilde{a} +1)(\widetilde{b}+\widetilde{b}')}\, Y^{b'} \frac{\partial x^{a}}{\partial x^{a'}}\frac{\partial x^{b}}{\partial x^{b'}} \, \Gamma_{ba}^{\,\,\, \,c} \frac{\partial x^{d'}}{\partial x^{c}}  \right)\frac{\partial}{\partial x^{d'}}\, ,\\
& = (-1)^{\widetilde{X} + \widetilde{a}'}\, X^{a'}\left (  \rho_{a'}^{\,\,b'}\frac{\partial Y^{d'}}{\partial x^{b'}} + (-1)^{(\widetilde{a}' +1)(\widetilde{Y} + \widetilde{b}')} \, Y^{b'} \, \Gamma_{b'a'}^{\,\,\,\,\, \,d'} \right) \frac{\partial}{\partial x^{d'}}.
\end{align*}
In the second term of the third line we have relabelled some of the contracted indices. Comparing the primed and unprimed coefficients established the required transformation rules. \\
\end{proof}
Naturally, and almost be definition, the odd endomorphism $\rho$ is a tensor of type $(1,1)$. The \emph{odd Chirstoffel symbols} $\Gamma_{ba}^{\,\,\, \,\,c}$ transform in almost the same was as their classical counterparts, as completely expected.

\subsection{The Torsion and Curvature}
We now proceed to generalise the notion of torsion and  curvature to odd quasi-connections.  The warning here is that the torsion  and  curvature will \emph{not}, in general, be tensors.  We have to be content, for the moment, with multi-linear maps (as vector spaces) in the definitions of torsion and curvature.
\begin{definition}\label{Def:Torsion}
Let $(\nabla, \, \rho)$ be an odd quasi-connection on a supermanifold $M$. The \emph{torsion}  of   $(\nabla, \, \rho)$ is is defined as the bi-linear map
$$T: \Vect(M) \times \Vect(M) \longrightarrow   \Vect(M)$$
given by
$$T(X,Y) := \nabla_X Y + (-1)^{\widetilde{X} \widetilde{Y}} \, \nabla_Y X + (-1)^{\widetilde{X}} \, \rho \big([\rho(X), \rho(Y)] \big),$$
for all (homogeneous) $X$ and $Y \in \Vect(M)$.
\end{definition}
It is easy to see that the torsion satisfies
\begin{enumerate}
\item $\widetilde{T(X,Y)} =  \widetilde{X} + \widetilde{Y} +1$, and
\item $T(X,Y) = (-1)^{\widetilde{X} \widetilde{Y}} \, T(Y,X)$,
\end{enumerate}
for all (homogeneous) $X$ and $Y \in \Vect(M)$.
\begin{definition}\label{Def:Torsionless}
An odd quasi-connection $(\nabla, \, \rho)$ is said to be \emph{torsionless} or \emph{torsion-free} if its associated torsion is the zero map.
\end{definition}
\begin{definition}\label{Def:RieCurv}
Let $(\nabla, \, \rho)$ be an odd quasi-connection on a supermanifold. The  curvature of  $(\nabla, \, \rho)$ is defined as the multi-linear map
$$R : \Vect(M) \times \Vect(M) \times \Vect(M) \longrightarrow \Vect(M)$$
given by
$$(X,Y,Z) \mapsto R(X,Y)Z := [\nabla_X, \nabla_Y]Z - \nabla_{\rho[\rho(X), \rho(Y)]}Z,$$
for all homogeneous $X,Y$ and $Z \in \Vect(M)$.
\end{definition}
It is easy to check that the curvature satisfies
\begin{enumerate}
\item $\widetilde{R(X,Y)Z} =  \widetilde{X} + \widetilde{Y} +\widetilde{Z}$,
\item $R(X,Y)Z = -(-1)^{(\widetilde{X}+1)( \widetilde{Y}+1)} \, R(Y,X)Z$, and
\item $R(X,Y)f\,Z = (-1)^{(\widetilde{X}+ \widetilde{Y}) \, \widetilde{f}} \,f\, R(Y,X)Z $,
\end{enumerate}
for all (homogeneous) $X,Y$ and $Z \in \Vect(M)$ and $f \in C^\infty(M)$. \par
We observe that, just as in the standard case of affine connections on a (super)manifold, that for any fixed pair $(X,Y)$ of homogeneous vector fields, we have a linear map (in the sense of a $C^\infty(M)$-module)
$$R(X,Y) : \Vect(M) \longrightarrow \Vect(M)$$
of Grassmann parity $ \widetilde{X} + \widetilde{Y}$.
\begin{definition}\label{Def:Flat}
An odd quasi-connection $(\nabla, \, \rho)$ is said to be \emph{flat} if its associated  curvature is the zero map.
\end{definition}
 We will refer to an odd quasi-connection whose odd endomorphism $\rho: \Vect(M) \rightarrow \Vect(M)$ is the zero map as an  \emph{odd banal quasi-connection}. Note that the third defining property of an odd quasi-connection reduces to
$\nabla_X f \,Y =  (-1)^{(\widetilde{X}+1) \widetilde{f}} \, f \, \nabla_X Y $. Thus, odd banal quasi-connections are precisely  odd tensors of type $(1,2)$.
\begin{proposition}\label{Prop:DifferBanal}
Let $(\nabla, \rho)$ and $(\nabla', \rho)$ be a pair of odd quasi-connections on a supermanifold $M$ with the same odd endomorphism $\rho$. Then the difference of the two odd quasi-connections is an odd banal quasi-connection.
\end{proposition}
\begin{proof}
We just need to check the following
$$(\nabla_X - \nabla'_X)f Y = \rho(X)f Y - \rho(X)f Y + (-1)^{(\widetilde{X}+1) \widetilde{f}} \,f \,(\nabla_X - \nabla'_X) Y = (-1)^{(\widetilde{X}+1) \widetilde{f}} \,f \,(\nabla_X - \nabla'_X) Y,$$
which is exactly the definition of an odd Banal quasi-connection.
\end{proof}

 \begin{definition}
 A \emph{odd involutive quasi-connection} on a supermanifold $M$ is an odd quasi-connection $(\nabla, \rho)$ on $M$ such that the odd endomorphism $\rho: \Vect(M) \rightarrow \Vect(M)$ is an involution.
 \end{definition}
\begin{remark}
Using the  nomenclature first introduced by Manin \cite[page 219]{Manin:1997}, a supermanifold equipped with an odd involution on its  module of vector fields is  said to be a \emph{$\Pi$-symmetric supermanifold}. The analogy with supersymmetry should not be missed. The odd involution exchanges a (homogeneous) vector field with one of a different Grassmann parity and applied twice we recover the initial vector field.  This is a kind of ``supersymmetry''.
\end{remark}
\begin{theorem}\label{Thm:TorCurTensors}
Let $( \nabla, \, \rho)$ be an odd quasi-connection on a supermanifold $M$. Let us assume that the associated torsion and  curvature are not both  zero maps.  The  torsion and  curvature of $( \nabla, \, \rho)$ are tensors on $M$ if and only if the odd quasi-connection is either banal or involutive.
\end{theorem}
\begin{proof}
We proceed to prove the theorem by checking the anomalous terms of the tensorial property of the torsion and  curvature. \par
\begin{itemize}
\item \textbf{Torsion} As $T$ is symmetric, it suffices to check the tensorial property  for one of the arguments. Thus we consider
\begin{align*}
T(X,fY) & = \nabla_X f\, Y + (-1)^{\widetilde{X} \widetilde{Y} + \widetilde{X} \widetilde{f}} \, \nabla_{fY} X + (-1)^{\widetilde{X}} \, \rho[\rho(X), \rho(fY)]\\
 &= \rho(X)f \, Y +  (-1)^{(\widetilde{X}+1)\widetilde{f}}\, f\, \nabla_X Y + (-1)^{\widetilde{X}} \, \rho[\rho(X),(-1)^{\widetilde{f}} f \rho(Y)]  \\
 & = \rho(X)f \, Y +  (-1)^{(\widetilde{X}+1)\widetilde{f}}\, f\, \nabla_X Y +  (-1)^{\widetilde{X}} \, \rho\big ( (-1)^{\widetilde{f}}\,\rho(X)f \, \rho(Y) + (-1)^{\widetilde{X}\widetilde{f}}\,f \, [\rho(X), \rho(Y)]  \big)\\
 & = (-1)^{(\widetilde{X}+1)\widetilde{f}}\, f\, T(X,Y) + \rho(X)f \, (Y - \rho(\rho(Y))).
\end{align*}
\item\textbf{Curvature} Due to the symmetry of $R(X,Y)$ it is sufficient to check the tensorial property for one of the arguments. Thus we consider
\begin{align*}
R(X, fY)Z & = [\nabla_X, \nabla_{fY}]Z - \nabla_{\rho[\rho(X), \rho(fY)]}Z\\
& = \nabla_X\big((-1)^{\widetilde{f}} f \, \nabla_Y Z \big ) - (-1)^{(\widetilde{X}+1)(\widetilde{Y}+1) + \widetilde{X} \widetilde{f}} \, f \, \nabla_Y \nabla_X Z - \nabla_{\rho[\rho(X),(-1)^{\widetilde{f}} f\rho(Y)]}Z\\
&=  (-1)^{\widetilde{f}} \, \rho(X)f \, \nabla_Y Z
 + (-1)^{\widetilde{X} \widetilde{f}} \, f \nabla_X \nabla_Y Z -(-1)^{(\widetilde{X}+1)(\widetilde{Y} +1) + \widetilde{X}\widetilde{f}} \, f \, \nabla_Y \nabla_X Z  \\
 & - \nabla_{\rho \big ( (-1)^{\widetilde{f}} \, \rho(X)f \, \rho(Y) + (-1)^{\widetilde{f}\, \widetilde{X}}\, f \, [\rho(X), \rho(Y)] \big )}Z\\
 & = (-1)^{\widetilde{f} \widetilde{X}}\, f \, R(X,Y)Z + (-1)^{\widetilde{f}}\, \rho(X)f \, \nabla_{\left(Y - \rho( \rho(Y))\right) }Z.
  \end{align*}
\end{itemize}
 In both cases the anomalies vanish if either $\rho =0$ or $\rho^2 = \Id$. The only if follows as the anomalous terms must vanish for all arbitrary vector fields and functions.
\end{proof}
\begin{remark}
It is clear that if both the torsion and  curvature are zero, then they are trivially tensors. So, we exclude this from our considerations.
\end{remark}

\subsection{Odd Connections as Odd Involutive Quasi-Connections}
From the previous subsection, it is clear that odd banal quasi-connections and odd involutive  quasi-connections have a rather privileged r\^{o}le in the theory of odd quasi-connections in that their torsion and curvature are geometric objects, i.e., they are tensors.  However, the banal case is  not so interesting and hence the name (which we hijacked from \cite{Etayo:1993}). We will focus attention for the remainder of this paper to odd involutive quasi-connections. In light of this, we will change nomenclature slightly.

\medskip

\noindent \textbf{Nomenclature:} By \emph{odd connection}, we explicitly mean an odd involutive quasi-connection.
\begin{remark}
The existence of an odd involution places heavy restrictions on the supermanifold. In particular, we must have a $n|n$-dimensional supermanifold such that the $C^\infty(M)$-module of vector fields admits a generating set consisting of $n$ even and $n$ odd vector fields.
\end{remark}
 \begin{example}[The canonical odd connection on $\R^{n|n}$]\label{Exp:CanOddCon}
 Consider the linear supermanifold $\R^{n|n}$ equipped with global coordinates $(x^a, \,\zx^b)$, of Grassmann parity $0$ and $1$, respectively. We define the canonical odd involution by its action on the partial derivatives, i.e.,
\begin{align*}
\rho \left( \frac{\partial}{\partial x^a} \right) =  \frac{\partial}{\partial \zx^a}, && \rho \left( \frac{\partial}{\partial \zx^b} \right) =  \frac{\partial}{\partial x^b}\,.
\end{align*}
We decompose any homogeneous vector field as
$$X = X^a(x,\zx)\frac{\partial}{\partial x^a} + \bar{X}^a(x,\zx)\frac{\partial}{\partial \zx^a},$$
so that the canonical odd connection on $\R^{n|n}$ is given by
$$\nabla_X Y:= ( -1)^{\widetilde{X}}    \left ( X^a \frac{\partial Y^b}{\partial \zx^a} -  \bar{X}^a \frac{\partial Y^b}{ \partial x^a}\right)\frac{\partial}{\partial x^b} +  ( -1)^{\widetilde{X}}    \left ( X^a \frac{\partial \bar{Y}^b}{\partial \zx^a} -  \bar{X}^a \frac{\partial \bar{Y}^b}{ \partial x^a}\right)\frac{\partial}{\partial \zx^b}\,. $$
\end{example}
 \begin{example}\label{Exp:1dSUSY}
 Consider the linear supermanifold $\R^{1|1}$ which we equip with global coordinates $(t, \theta)$. We pick as a global basis for the vector fields  $\{ P := \partial_t , \, D := \partial_\theta - \theta \partial_t  \}$. The reader should immediately recognise this as the SUSY structure on $\R^{1|1}$, i.e., $\mathcal{D} =\Span\{ D \}$ is a maximally non-integrable distribution of rank $0|1$.  We can then define the odd involution as  $\rho(P) = D$ and $\rho(D) =P$. Then, following Example \ref{Exp:CanOddCon}, we write $X = X^P(t,\theta)P + X^{D}(t,\theta)D$ and define an odd connection as
 $$\nabla_X Y:= ( -1)^{\widetilde{X}}    \left ( X^P D(Y^P) -  X^DP(Y^P) \right)P +  ( -1)^{\widetilde{X}}    \left ( X^P D(Y^D) -  X^DP(Y^D) \right)D\,. $$
 \end{example}

We now want to examine if the set of all odd connections on a supermanifold has the structure of an affine space. Let $(\nabla, \, \rho)$ and $(\nabla', \, \rho')$ be odd connections on a supermanifold $M$. The only modification to the proof of Proposition \ref{Prop:AffineModule} is to check that $t \, \rho + (1-t)\, \rho'$ is itself an involution.  A direct calculation gives, for an arbitrary $X \in \Vect(M)$
 $$(t \, \rho + (1-t)\, \rho')^2 X = X+  (t - t^2)\big( \rho(\rho'(X)) +\rho'(\rho(X)) - 2X \big).$$
 For the affine combination to be an involution for all $t \in \R$ we require the one-dimensional Clifford--Dirac relation
 $$[\rho, \rho'] = 2\cdot \Id_{\Vect(M)}$$
 to hold. Here the bracket is the $\Z_2$-graded commutator bracket,  i.e., an anticommutator in the language of physics. The dimension of a one-dimensional Clifford algebra (as a vector space) is two.
\begin{definition}
A pair of odd connections $(\nabla, \,\rho)$ and $(\nabla', \,\rho')$ are said to be \emph{compatible} if and only if they satisfy the Clifford--Dirac relation
 $$[\rho, \rho'] = 2\cdot \Id_{\Vect(M)}.$$
\end{definition}
With the above observations and definition in place, we have the following result.
\begin{theorem}
The set of all pairwise compatible odd connections on  a supermanifold has the structure of an affine space.
\end{theorem}
\begin{remark}
It is clear that the set of odd connections on $M$ cannot be a $C^\infty_0(M)$-module in the same way as general odd quasi-connections. The involutive property of $\rho$ is not preserved by multiplication by an even function.
\end{remark}

\subsection{Induced Odd Connections}
It turns out that, in much the same way as with quasi-connections, odd connections and affine connections are not completely separate notions.
\begin{proposition}\label{Prop:CanGen}
Let $M$ be an $n|n$-dimensional supermanifold equipped with an affine connection $\bar{\nabla}$ and an odd involution $\rho: \Vect(M) \rightarrow \Vect(M)$. Then $(\nabla, \rho)$, were
$$\nabla:= \bar{\nabla} \circ (\rho, \, \Id_{\Vect(M)}),$$
is an odd connection on $M$.
\end{proposition}
\begin{proof}
We just need to check the defining properties of an odd quasi-connection. The Grassmann parity is clear as an affine connection is an even map. $\nabla_{fX} Y =  \bar{\nabla}_{\rho(fX)}Y = (-1)^{\widetilde{f}} f \,\bar{\nabla}_{\rho(X)}Y = (-1)^{\widetilde{f}} f \,\nabla_X Y  $, establishes the second condition. The third condition similarly follows from a short calculation $\nabla_X fY = \bar{\nabla}_{\rho(X)} fY = \rho(X)f \, Y + (-1)^{(\widetilde{X} +1)\widetilde{f}} \, f \, \bar{\nabla}_{\rho(X)}Y= \rho(X)f \, Y + (-1)^{(\widetilde{X} +1)\widetilde{f}} \, f \, \nabla_{X}Y$.
\end{proof}
\begin{remark}
Clearly, if we relax the condition that $\rho$ is an involution we arrive at a general odd quasi-connection.
\end{remark}
\begin{definition}
Let $\bar{\nabla}$ be an affine connection on $M$. An odd connection $(\nabla, \rho )$ on $M$ is said to be \emph{canonically generated by} $\bar{\nabla}$ if and only if
$$\nabla= \bar{\nabla} \circ (\rho, \, \Id_{\Vect(M)}).$$
\end{definition}
\begin{example}
The canonical odd connection on $\R^{n|n}$ (see Example \ref{Exp:CanOddCon}) is an example of a canonically induced odd connection where the affine connection is the standard connection on $\R^{n|n}$ and the odd involution is the canonical one.
\end{example}
Directly from Proposition \ref{Prop:DifferBanal} we have the following result.
\begin{proposition}\label{Prop:ConAffDifBan}
Let $(\nabla, \, \rho)$ be an odd connection and $\bar{\nabla}$ be an arbitrary affine connection both the same supermanifold $M$. Then
$$B := \nabla -  \bar{\nabla} \circ (\rho, \, \Id_{\Vect(M)}) $$
is an odd banal quasi-connection.
\end{proposition}
\begin{proposition}\label{Prop:ConGenII}
An odd connection $(\nabla, \rho)$ on a supermanifold $M$ is canonically generated by the affine connection $\bar{\nabla} = \nabla \circ (\rho, \, \Id_{\Vect(M)})$.
\end{proposition}
\begin{proof}
First, we need to show that  $\bar{\nabla}$ is an affine connection. A quick calculation similar to that used in the proof of Proposition \ref{Prop:CanGen} shows this is the case.  Using Proposition \ref{Prop:ConAffDifBan} and the fact that $\rho$ is an involution we observe that
$$B(X,Y) = \nabla_X Y - \bar{\nabla}_{\rho(X)}Y = \nabla_X Y - \nabla_{\rho(\rho(X))}Y  =  \nabla_X Y - \nabla_X Y =0,$$
for arbitrary $X$ and $Y \in \Vect(M)$.  This implies the result.
\end{proof}

\begin{statement}
Given an affine connection and an odd involution  one can canonically construct an odd connection. Moreover, any odd connection has a decomposition into an induced odd connection (with respect to any chosen affine connection) and an odd banal connection, i.e., an odd tensor of type $(1,2)$.
\end{statement}

\subsection{The Algebraic Bianchi Identity}
We further justify our definition of an odd connection and, in particular, the definitions of the torsion and curvature. We view the classical first or algebraic Bianchi identity as a compatibility condition between the  curvature and the torsion. Thus, there should, if our notions are consistent, be some similar compatibility for the case of odd connections.
\begin{theorem}\label{Thm:Bianchi1}
Let $(\nabla, \, \rho)$ be an odd connection on a supermanifold. The associated torsion and  curvature tensors $T$ and $R$, respectively, satisfy the following generalisation of the first (or algebraic) Bianchi identity,
\begin{align*}
 &(-1)^{\widetilde{X}(\widetilde{Z} +1)}\,R(X,Y)Z + (-1)^{\widetilde{Y}(\widetilde{X} +1)}\,R(Y,Z)X  + (-1)^{\widetilde{Z}(\widetilde{Y} +1)}\,R(Z,X)Y \\
 &=(-1)^{\widetilde{X}(\widetilde{Z} +1)}\, \nabla_X\big( T(Y,Z)\big) +(-1)^{\widetilde{Y}(\widetilde{X} +1)}\, \nabla_Y\big( T(Z,X)\big) +(-1)^{\widetilde{Z}(\widetilde{Y} +1)}\, \nabla_Z\big( T(X,Y)\big)\\
 &-(-1)^{\widetilde{X}(\widetilde{Z} +1) + \widetilde{Y}}\, T\big( X, \rho[\rho(Y), \rho(Z)]  \big)-(-1)^{\widetilde{Y}(\widetilde{X} +1) + \widetilde{Z}}\, T\big( Y, \rho[\rho(Z), \rho(X)]  \big)-(-1)^{\widetilde{Z}(\widetilde{Y} +1) + \widetilde{X}}\, T\big( Z, \rho[\rho(X), \rho(Y)]  \big),
 \end{align*}
for all $X,Y$ and $Z \in \Vect(M)$.
\end{theorem}
\begin{corollary}
If the odd connection $(\nabla, \, \rho)$ in question is torsion-free, i.e., the torsion tensor vanishes,  then the first or algebraic Bianchi identity is
$$(-1)^{\widetilde{X}(\widetilde{Z} +1)}\,R(X,Y)Z + (-1)^{\widetilde{Y}(\widetilde{X} +1)}\,R(Y,Z)X  + (-1)^{\widetilde{Z}(\widetilde{Y} +1)}\,R(Z,X)Y = 0\,,$$
for all $X,Y$ and $Z \in \Vect(M)$.
\end{corollary}

The proof of Theorem \ref{Thm:Bianchi1}  follows from a direct, but laborious computation along the same lines as the classical proof. We defer details  to   Appendix \ref{App:ProofAlgBian}.

\subsection{Extending the Odd Connection to Tensor Fields}
The space of tensor fields, understood as representations of $\textnormal{GL}(n|m)$, i.e., the structure group of the tangent bundle of a $n|m$-dimensional supermanifold, is not exhausted by covariant, contravariant and mixed tensors.  One needs to include Berezin densities to ``close'' the theory and these cannot be constructed from `na\"{\i}ve' tensors.  In this subsection, we will concentrate on  $(p,q)$-tensors and how to extend the odd connection to act upon them. This is, of course, done via the Leibniz rule, just as in the classical setting. We do this in the logical steps of first extending to functions and one-forms before deducing what happens to mixed tensor fields.\par
For functions, we take the natural modification of the classical definition, i.e., for any and all $f \in C^\infty(M)$ we define
\begin{equation}\label{Eqn:OddConFun}
\nabla_X f:= \rho(X).
\end{equation}
Moving on to one-forms, we use the coordinate basis and locally write $\alpha = \delta x^a \, \alpha_a(x)$, where $\widetilde{\delta x^a} = \widetilde{a}$. That is, we are using ``even'' one-forms, see \cite[Appendix A.1]{Voronov:2016} for details.  Under general coordinate changes $x^{a'} = x^{a'}(x)$, the differentials transform as
 $$\delta x^{a'} = \delta x^a \left (\frac{\partial x^{a'}}{\partial x^a}\right).$$
The duality condition between the coordinate basis of vector fields and one-forms is, as in the classical setting,  $\langle \partial_a , \delta x^b\rangle = \delta_a^{\;\; b}$. Thus the invariant pairing between vector fields and one-forms is locally expressed as $\langle X, \alpha\rangle = X^a\alpha_a (x) \in C^\infty(M)$.   We then define the action of the odd covariant derivative on a one-form via this pairing and the Leibniz rule,
$$\nabla_X \langle Y, \alpha\rangle = \rho(X) \langle Y, \alpha\rangle =\langle \nabla_X Y, \alpha\rangle + (-1)^{(\widetilde{X}+1)\widetilde{Y}}\, \langle Y, \nabla_X \alpha\rangle.$$
Thus, $\nabla_X \alpha$ is completely determined by
$$ \langle Y, \nabla_X \alpha\rangle = (-1)^{(\widetilde{X}+1)\widetilde{Y}}\, \big(  \rho(X) \langle Y, \alpha\rangle - \langle \nabla_X Y, \alpha\rangle   \big).$$
Using \eqref{Eqn:LocForOddQuaiCon} and the invariant paring we see that
\begin{align*}
Y^a(\nabla_X \alpha)_a & = (-1)^{(\widetilde{X} +1)\widetilde{Y}} \left((-1)^{\widetilde{X} + \widetilde{b}} X^b \rho_b^{\;\; c }\,\frac{\partial}{\partial x^c}(Y^a\alpha_a) - (-1)^{\widetilde{X} + \widetilde{b}}\, \frac{\partial Y^a}{\partial x^c} \,\alpha_a - (-1)^{\widetilde{X} + \widetilde{b} +(\widetilde{b} +1)(\widetilde{Y} + \widetilde{a})}\, X^b Y^a \Gamma_{ab}^{\; \; \; c} \alpha_c \right)\\
& = Y^a\left((-1)^{\widetilde{X}(\widetilde{a} +1) + \widetilde{a} + \widetilde{b}} \,\left( X^b \rho_b^{\;\; c }\,\frac{\partial\alpha_a }{\partial x^c}  -  X^b  \Gamma_{ab}^{\; \; \; c} \alpha_c\right) \right).
\end{align*}
This implies that locally and using the coordinate basis, the odd connection acting on a one form is given by
\begin{equation}\label{Eqn:OddConOneForm}
\nabla_X \alpha :=    (-1)^{\widetilde{X}(\widetilde{a} +1) + \widetilde{a} + \widetilde{b}} \, \delta x^a  \,\left(   X^b \rho_b^{\;\; c }\,\frac{\partial\alpha_a }{\partial x^c}  -  X^b  \Gamma_{ab}^{\; \; \; c} \alpha_c \right).
\end{equation}
We are now in a position to describe what happens to more general mixed tensors. A $(p,q)$-tensor field is a ($\Z_2$-graded homogeneous)  $C^\infty(M)$-multilinear map
$$ T : \otimes^p \, \Vect(M) ~  \otimes^q \,\Omega^1(M) \longrightarrow  C^\infty(M),$$
where we have, of course, employed the $\Z_2$-graded tensor product over the global functions on $M$. Note that we have assumed no symmetry in this definition  and that we are using the ``even'' conventions with the one-forms.  In terms of the coordinate basis, so locally,  we write
$$T = \delta x^{a_1} \delta x^{a_2} \cdots \delta x^{a_p} \,T_{a_p \cdots a_2 a_1}^{\;\;\;\;\;\;\;\;\;\;\; b_q \cdots b_2 b_1} \, \partial_{b_1}\partial_{b_2} \cdots \partial_{b_q},$$
where we neglect to write the tensor product explicitly.  It is straightforward to deduce that under a general coordinate transformation the components of a tensor transform as
$$T_{a'_p \cdots a'_2 a'_1}^{\;\;\;\;\;\;\;\;\;\;\;\; b'_q \cdots b'_2 b'_1} = (-1)^\chi \, \frac{\partial x^{a_1}}{\partial x^{a'_1}} \frac{\partial x^{a_2}}{\partial x^{a'_2}} \cdots \frac{\partial x^{a_p}}{\partial x^{a'_p}} \,T_{a_p \cdots a_2 a_1}^{\;\;\;\;\;\;\;\;\;\;\; b_q \cdots b_2 b_1} \,\frac{\partial x^{b'_1}}{\partial x^{b_1}} \frac{\partial x^{b'_2}}{\partial x^{b_2}} \cdots \frac{\partial x^{b'_q}}{\partial x^{b_q}}\,,$$
where the sign factor is given by
\begin{align*}
\chi &= (\widetilde{a}_1 + \widetilde{a}'_1 )(\widetilde{a}'_2 + \widetilde{a}'_3 + \cdots + \widetilde{a}'_p) +(\widetilde{a}_2 + \widetilde{a}'_2 )(\widetilde{a}'_3 + \cdots + \widetilde{a}'_p) + \cdots + (\widetilde{a}_{p-1} + \widetilde{a}'_{p-1} ) \widetilde{a}'_p\\
& +  (\widetilde{b}_2 + \widetilde{b}'_2 )\widetilde{b}'_1 + (\widetilde{b}_3 + \widetilde{b}'_3 )(\widetilde{b}'_1 +\widetilde{b}'_2 ) + \cdots + (\widetilde{b}_q + \widetilde{b}'_q )(\widetilde{b}'_1 +\widetilde{b}'_2 + \cdots + \widetilde{b}'_{q-1}).
\end{align*}

%
\begin{warning}
There are plenty of other conventions in the literature with regards to the ordering and position of indices of tensor fields. Note that we put the components of the tensor in the middle and this will effect various sign factors.
\end{warning}
Just as in the classical setting, we define the action of an odd connection on a mixed tensor field via the Leibniz rule. After a little rearranging one obtains the following definition.
\begin{definition}\label{Def:ExtOddConTen}
Let $(\nabla, \, \rho)$ be an odd connection on a supermanifold $M$. Furthermore, let $T$ be a $(p,q)$-tensor field on $M$, and $\{Y_1, \cdots ,Y_p \}$  and $\{ \alpha^1, \cdots , \alpha^q \}$ be collections of (homogeneous) arbitrary vector fields and one-forms, respectively. Then the \emph{odd covariant derivative of  $T$ in the direction of $X \in \Vect(M)$} is defined as
 \begin{align}
 \nonumber
 \big(\nabla_X T\big)(Y_1, Y_2, \cdots , Y_p \, ;\, \alpha^1, \alpha^2, \cdots , \alpha^q) & = (-1)^{(\widetilde{X}+1)(\widetilde{Y}_1 +\widetilde{Y}_2 + \cdots + \widetilde{Y}_p)}\,  \rho(X)\big(T(Y_1, Y_2, \cdots , Y_p \, ;\, \alpha^1, \alpha^2, \cdots , \alpha^q)\big)\\ \nonumber
 & -  (-1)^{(\widetilde{X}+1)(\widetilde{Y}_1 +\widetilde{Y}_2 + \cdots + \widetilde{Y}_p)} \, T(\nabla_X Y_1, Y_2, \cdots , Y_p \, ;\, \alpha^1, \alpha^2, \cdots , \alpha^q)\\ \nonumber
 & -  (-1)^{(\widetilde{X}+1)(\widetilde{Y}_2+ \widetilde{Y}_3 + \cdots + \widetilde{Y}_p)} \, T( Y_1,\nabla_X  Y_2, \cdots , Y_p \, ;\, \alpha^1, \alpha^2, \cdots , \alpha^q)\\ \nonumber
 & \vdots \\\nonumber
  & -  (-1)^{(\widetilde{X}+1)\widetilde{Y}_p} \, T( Y_1,  Y_2, \cdots , \nabla_X Y_p \, ;\, \alpha^1, \alpha^2, \cdots , \alpha^q)\\\nonumber
  & -  (-1)^{(\widetilde{X}+1)\widetilde{T}} \, T( Y_1,  Y_2, \cdots ,  Y_p \, ;\, \nabla_X \alpha^1, \alpha^2, \cdots , \alpha^q)\\\nonumber
   & -  (-1)^{(\widetilde{X}+1)(\widetilde{T}+ \widetilde{\alpha}^1)} \, T( Y_1,  Y_2, \cdots ,  Y_p \, ;\,  \alpha^1, \nabla_X \alpha^2, \cdots , \alpha^q)\\ \nonumber
   & \vdots\\
   & -  (-1)^{(\widetilde{X}+1)(\widetilde{T}+ \widetilde{\alpha}^1 + \widetilde{\alpha}^2 + \cdots + \widetilde{\alpha}^{q-1})} \, T( Y_1,  Y_2, \cdots ,  Y_p \, ;\,  \alpha^1, \alpha^2, \cdots ,  \nabla_X \alpha^q).
 \end{align}
 \end{definition}
As a specific example, consider a $(1,1)$-tensor, written locally as $T(Y;\alpha) = Y^aT_a^{\;\; b}\alpha_b$. The directly applying Definition \ref{Def:ExtOddConTen} together with \eqref{Eqn:LocForOddQuaiCon} and  \eqref{Eqn:OddConOneForm} we see that
\begin{align*}
\left( \nabla_X T\right)(Y;\alpha) & =(-1)^{(\widetilde{X}+1)\widetilde{Y}}\, \rho(X)\big( Y^aT_a^{\;\; b}\alpha_b \big ) - (-1)^{(\widetilde{X}+1)\widetilde{Y}}\,\big( \nabla_X Y\big)^a T_a^{\;\; b}\alpha_b - (-1)^{(\widetilde{X}+1)\widetilde{T}}\,Y^a T_a^{\;\; b}\big( \nabla_X \alpha\big)_b\\
& = (-1)^{(\widetilde{X}+ 1)\widetilde{Y}}\, X^c \rho_c^{\;\; d}\frac{\partial Y^a}{\partial x^d} T_a^{\;\; b} \alpha_b + (-1)^{\widetilde{X}(\widetilde{a}+ 1) \widetilde{a} + \widetilde{c}}\, Y^a X^c \rho_c^{\;\; d} \frac{\partial T_a^{\;\; b}}{\partial x^d} \alpha_b\\
& +  (-1)^{\widetilde{X}(\widetilde{T}+ \widetilde{b}+ 1) \widetilde{c} + \widetilde{b} + \widetilde{T}}\, Y^a T_a^{\;\; b} X^c \rho_c^{\;\; d} \frac{\partial  \alpha_b}{\partial x^d} -(-1)^{(\widetilde{X}+ 1)\widetilde{Y}}\, X^c \rho_c^{\;\; d}\frac{\partial Y^a}{\partial x^d} T_a^{\;\; b} \alpha_b \\
& - (-1)^{\widetilde{X}(\widetilde{a} +1)+ \widetilde{a} + \widetilde{c}}\, Y^a X^c \Gamma_{ac}^{\;\;\;d }T_b^{\;\; b}\alpha_b - (-1)^{\widetilde{X}(\widetilde{T}+ \widetilde{b}+ 1) \widetilde{c} + \widetilde{b} + \widetilde{T}}\, Y^a T_a^{\;\; b} X^c \rho_c^{\;\; d} \frac{\partial  \alpha_b}{\partial x^d}\\
&- (-1)^{\widetilde{X}(\widetilde{a} +1) + \widetilde{T} + \widetilde{c}(\widetilde{a} + \widetilde{d}+1) + \widetilde{d}} \, Y^a X^cT_a^{\;\; d}\Gamma_{d c}^{\;\;\; b}\alpha_b\\
& = (-1)^{\widetilde{X}(\widetilde{a}+1)} \, Y^a X^c \left((-1)^{\widetilde{a}+\widetilde{c}} \, \rho_c^{\;\; d} \frac{\partial T_a^{\;\; b}}{\partial x ^d} -  (-1)^{\widetilde{a}+\widetilde{c}} \, \Gamma_{ac}^{\;\;\; d}T_d^{\;\; b} + (-1)^{\widetilde{T} + \widetilde{c}(\widetilde{a} + \widetilde{d}+1) + \widetilde{d}} \, T_a^{\;\; d} \Gamma_{dc}^{\;\; \; b}\right)\alpha_b \, .
\end{align*}
\begin{definition}\label{Def:OddConCompG}
Let $G$ be a rank-2 covariant tensor  on $M$ (no symmetry or non-degeneracy is assumed). Then an odd connection  $(\nabla, \, \rho)$ is \emph{compatible with $G$} if and only if $\big(\nabla_X G\big)(Y,Z)=0$ for all (homogeneous) $X,Y$ and $Z \in \Vect(M)$.
\end{definition}
Using Definition \ref{Def:ExtOddConTen} it is clear that the compatibility condition can be written as
\begin{equation}\label{Eqn:OddConCompG}
\rho(X)\big( G(Y,Z)\big) = G(\nabla_X Y,Z) + (-1)^{(\widetilde{X}+1)\widetilde{Y}}\, G(Y, \nabla_X Z)\,.
\end{equation}

\subsection{Odd Divergence Operators}
Divergence operators in supergeometry are even maps $ \Vect(M) \rightarrow C^\infty(M)$ that can be defined in terms of a Berezin volume or an affine connection.  The two approaches are, of course, tightly related, just as they are in the classical setting (see \cite{Khudaverdian:2004b,Khudaverdian:2013, Kosmann-Schwarzbach:2002}).  For the definition of an odd divergence operator we are forced to generalise the definition of a divergence operator in terms of an affine connection.
\begin{definition}\label{Def:OddDivOp}
Let $(\nabla, \, \rho)$ be an odd connection on a supermanifold $M$. The associated \emph{odd divergence operator} is the odd map
$$\textnormal{Div}_{(\nabla,\rho)} : ~\Vect(M) \longrightarrow C^\infty(M)$$
defined in local coordinates as
$$\textnormal{Div}_{(\nabla,\rho)} \, X := (-1)^{\widetilde{a}(\widetilde{X}+1)}\big( \nabla_a X \big)^a = (-1)^{ \widetilde{a}(\widetilde{X} +1)}\, \left (  \rho_a^{\,\,b}\frac{\partial X^a}{\partial x^b} + (-1)^{(\widetilde{a} +1)(\widetilde{X} + \widetilde{b})} \, X^b \, \Gamma_{ba}^{\,\,\, \,a} \right),$$
for any and all homogeneous $X\in \Vect(M)$.
\end{definition}
We must check that an odd divergence operator is well-defined, i.e., does not depend on the coordinates used. From Proposition \ref{Prop:TransOddQuasiCon} we see that
$$(-1)^{\widetilde{a}'}\, \big( \nabla_{a'}X\big)^{b'} = (-1)^{\widetilde{a}}\, \left(\frac{\partial x^{a}}{\partial x^{a'}}\right) \big( \nabla_{a}X\big)^{b}\left(\frac{\partial x^{b'}}{\partial x^b}\right).$$
This implies the following
\begin{align*}
(-1)^{\widetilde{a}'(\widetilde{X} +1)}\, \big( \nabla_{a'}X\big)^{a'} & = (-1)^{ \widetilde{a}' \, \widetilde{X} + \widetilde{a}}\, \left(\frac{\partial x^{a}}{\partial x^{a'}}\right) \big( \nabla_{a}X\big)^{b}\left(\frac{\partial x^{a'}}{\partial x^b}\right)\\
&= (-1)^{ \widetilde{a}(\widetilde{X}+1)}\,  \big( \nabla_{a}X\big)^{b}\left(\frac{\partial x^{a'}}{\partial x^b}\right)\left(\frac{\partial x^{a}}{\partial x^{a'}}\right)\\
& = (-1)^{ \widetilde{a}(\widetilde{X}+1)}\,  \big( \nabla_{a}X\big)^a\,,
\end{align*}
and so we conclude that the definition of an odd divergence operator is sound.
\begin{remark}
The  definition of an odd divergence operator generalises to odd quasi-connections with no problem.
\end{remark}
\begin{proposition}\label{Prop:OddDivOpPropeties}
 Let $(\nabla, \, \rho)$ be an odd connection on a supermanifold $M$ and furthermore, let $\textnormal{Div}_{(\nabla,\rho)}$ be the associated odd divergence operator.  The following properties hold.
 \begin{enumerate}
 \item $\textnormal{Div}_{(\nabla,\rho)} (X+ \lambda\, Y) = \textnormal{Div}_{(\nabla,\rho)} X + \lambda \, \textnormal{Div}_{(\nabla,\rho)}  Y$,
 \item $\textnormal{Div}_{(\nabla,\rho)} (f\, X) = (-1)^{\widetilde{f}}\, f \, \textnormal{Div}_{(\nabla,\rho)} X  + (-1)^{\widetilde{X} \, \widetilde{f}}\, \rho(X)f$, \label{Eqn:FunPropDiv}
 \end{enumerate}
 for all $X$ and $Y \in \Vect(M)$ homogeneous and of the same degree,  $\lambda \in \R$ and $f \in C^\infty(M)$.
\end{proposition}
\begin{proof}
As we have fixed the odd connection, we will use the shorthand $\textnormal{Div}$ for the odd divergence operator.
\begin{enumerate}
\item This follows directly from the $\R$-linearity of  odd connections (see Definition \ref{Def:OddQuaCon}) and the local definition of the odd divergence. Specifically,
\begin{align*}
\textnormal{Div}(X + \lambda \, Y) &= (-1)^{\widetilde{a}(\widetilde{X}+1)}\big( \nabla_a X  + \lambda \, Y\big)^a
 = (-1)^{\widetilde{a}(\widetilde{X}+1)}\big( \nabla_a X \big)^a + (-1)^{\widetilde{a}(\widetilde{Y}+1)}\big( \nabla_a (\lambda \, Y) \big)^a\\
&= \textnormal{Div} X + \lambda \, \textnormal{Div}  Y.
\end{align*}
\item Similarity, this follows from the Leibniz rule for odd connections (see Definition \ref{Def:OddQuaCon}).
\begin{align*}
\textnormal{Div}(f\,X) & =  (-1)^{\widetilde{a}(\widetilde{X}+1) + \widetilde{a} \widetilde{f}}\big( \nabla_a f\,X \big)^a
= (-1)^{\widetilde{a}(\widetilde{X}+1) + \widetilde{a} \widetilde{f}} \, \rho_a^{\,\, b}\frac{\partial f}{\partial x^b}\, X^a + (-1)^{\widetilde{a}(\widetilde{X}+1) + \widetilde{f}}\, f \big(\nabla_a X \big)^a\\
&= (-1)^{\widetilde{a}(\widetilde{X}+1) + \widetilde{f}}\, f \big(\nabla_a X \big)^a +  (-1)^{\widetilde{X} \widetilde{f} + \widetilde{X} +\widetilde{a}}  \,X^a \rho_a^{\,\, b}\frac{\partial f}{\partial x^b}
= (-1)^{\widetilde{f}}\, f \, \textnormal{Div}  X  + (-1)^{\widetilde{X} \, \widetilde{f}}\, \rho(X)f.
\end{align*}
\end{enumerate}
\end{proof}
For non-homogeneous  vector fields we extend the definition of the odd divergence via linearity. Note that property \eqref{Eqn:FunPropDiv} of Proposition \ref{Prop:OddDivOpPropeties} is the odd generalisation of the defining property of any divergence operator.

\begin{example}\label{Exp:OddDiv1dSUSY}
Continuing with Example \ref{Exp:1dSUSY}, the odd divergence operator with respect to the odd connection defined by  the basis $\{P,D \}$ is
$$\textnormal{Div}\,X =  D(X^P) - (-1)^{\widetilde{X}} \, P(X^D)\,,$$
where, as before, $X = X^P P + X^D D$.
\end{example}

\subsection{On the Existence of Odd Connections}\label{SubSec:Existence}
Proposition  \ref{Prop:ConAffDifBan} tells us that up to an odd Banal connection any odd connection is a canonically induced odd connection with respect to any specified affine connection.  Thus, it is without great loss of generality to consider canonically induced odd connections when it comes to the question of the existence of odd connections.\par
\begin{lemma}\label{Lem:ExtCon}
The set of affine connections on a (smooth) supermanifold $M$ is non-empty.
\end{lemma}
\begin{proof}
This is a well-established fact and so we will only highlight the main argument. As we are dealing with real smooth supermanifolds partitions of unity always exists (see for example \cite[Lemma 3.1.7]{Leites:1980}). One can then amend the classical proof of the existence of affine connections on smooth manifolds to the setting of smooth supermanifolds.
\end{proof}
\begin{lemma}\label{Lem:ExtInv}
Let $G$ be a (smooth) Lie supergroup of dimension $n|n$. Then the set of odd involutions of the $C^\infty(G)$-module of vector fields on $G$ is non-empty.
\end{lemma}
\begin{proof}
Lie supergroups admit a global frame for the module of vector fields (you get the same result as for Lie groups, which you state as the tangent bundle being trivial). In our case, for a Lie supergroup $G$,  $\Vect(G) = \cO_G(|G|) \otimes \mathfrak{g}$, where $\mathfrak{g}$ is the Lie superalgebra of the supergroup (see for example \cite[Proposition 2.9]{Boyer:1991}). The Lie superalgebra    is of dimension $n|n$, therefore it admits  an odd involution. Then, all Lie supergroups of dimension $n|n$ can be equipped with an odd involution, i.e., the module of vector fields is $\Pi$-\emph{symmetric} in the language of Manin \cite{Manin:1997} and others \cite{Bouarroudj:2012}.
\end{proof}
\begin{theorem}\label{Thm:ExodConLieGrp}
Let $G$ be a $n|n$-dimensional Lie supergroup. Then the set of odd connections on $G$ is non-empty.
\end{theorem}
\begin{proof}
A direct consequence of Lemma \ref{Lem:ExtCon}, Lemma \ref{Lem:ExtInv} and Proposition \ref{Prop:CanGen}.
\end{proof}

\begin{example}
The Lie supergroups $GL(m|m)$, $SL(m |m+1)$, $Osp(2m|2m)$, and $Q(m)$ are of dimension $2m^2|2m^2$,  $2m(m+1)|2m(m+1)$, $4m^2|4 m^2$ and $m^2|m^2$, respectively. Thus, they each admit odd connections.
\end{example}

The previous theorem generalises to supermanifolds that admit a global frame for their vector fields, but do not necessarily have the underlying structure of a Lie supergroup. Recall that any $X \in \Vect(M)$ defines for any point $p \in |M|$     an induced derivation of sections of the stalk at $p$ of the structure sheaf, denoted $X_p \in \textnormal{Der}\,\cO_p$. We define $X_p := \textnormal{ev}_p \circ \epsilon \circ X:  \cO_p \rightarrow \R$ which is a linear map that satisfied the Leibniz rule
$$X_p(st) = X_p(s) \, (\epsilon\,t)(p) + (-1)^{\widetilde{X} \, \widetilde{s}} \, (\epsilon\, s) \, X_p(t), $$
where $\epsilon : \cO_p \rightarrow C^\infty_p$ is the algebra morphism induced by the sheaf morphisms $\epsilon_{-} : \cO_M(-) \rightarrow C^\infty_{|M|} (-)$.  The map $\textnormal{ev}_p: C^\infty_{|M|} \rightarrow \R$ is the standard evaluation map. It is customary to define $\sT_p M := \textnormal{Der}\,\cO_p$ as the tangent space at $p$. This is, of course, a super vector space and  for every $p$ we have an isomorphism $\sT_p M \cong \vec{\R}^{n|m}$. Recall that   a \emph{parallelisation} of a supermanifold $M$ is a set $\{X_1, X_2, \cdots, X_n \,; \,Y_1, Y_2, \cdots , Y_m   \}$ of $n|m$ vector fields such that for every $p \in |M|$ the set of induced derivations  $\{(X_1)_p, (X_2)_p, \cdots, (X_n)_p \,; \,(Y_1)_p, (Y_2)_p, \cdots , (Y_m)_p   \}$ is a basis of the tangent space $\sT_p M \cong \vec{\R}^{n|m}$. A supermanifold is called \emph{parallelisable} if it admits a parallelisation. A choice of parallelisation establishes the isomorphism of $C^\infty(M)$-modules $\Vect(M) \stackrel{\sim}{\rightarrow} C^\infty(M)\otimes \vec{\R}^{n|m}$.
\begin{theorem}\label{Thm:OddConPara}
The set of odd connections on a $n|n$-dimensional parallelisable supermanifold $M$ is non-empty.
\end{theorem}
\begin{proof}
By definition $n|n$-dimensional parallelisable supermanifolds admit a global frame consisting of $n|n$ vector fields and so the set of odd involutions is non-empty. For example, if we choose some parallelisation $\{ X_i \,; Y_j\}$, where $i,j = 1,2, \cdots n$, then we can define a canonical odd involution associated with this choice, i.e., $\rho(X_i) = Y_i$ and $\rho(Y_j) = X_j$ is an odd involution.  The existence of odd connections  then  follows from Lemma \ref{Lem:ExtCon} and Proposition \ref{Prop:CanGen}.
\end{proof}
\begin{remark}
On a $n|n$-dimensional parallelisable supermanifold we can construct an odd Riemannian metric by setting $g(X_i, X_j) = g(Y_i, Y_j) =0$ and $g(X_i, Y_j) = \delta_{ij}$. This suggests that odd Riemannian metrics are not as ``unnatural'' as one might at first think. Moreover, we will use this metric in  Subsection \ref{SubSec:OddWeitCon}.
\end{remark}

\subsection{The Odd Connection on Super-Minkowski Spacetime}
We will restrict attention to $d =4$ and $\mathcal N =1$ super-Minkowski spacetime, which we will denote as $\textnormal{SMink}^{4|4}$. We will comment on other dimensions and extended supersymmetries at the end of this subsection.  As a supermanifold  $\textnormal{SMink}^{4|4} = \R^{4|4}$ and comes equipped with global coordinates $(x^\mu, \theta_\alpha)$, where $x^\mu$ transforms under the Lorentz group as a vector and $\theta_\alpha$ transforms as a Majorana spinor (we will follow the conventions of \cite[Section 2.3]{Bruce:2019b}). We will work in the manifestly real setting and so the Lorentzian metric is $\textnormal{diag}(-1, +1,+1, +1)$. This allows us to use the real Majorana of the Clifford algebra $\mathcal{C}l(3,1)$.  \par
The SUSY structure on $\textnormal{SMink}^{4|4}$ is the maximally non-integrable distribution spanned by the SUSY covariant derivatives
$$D^\alpha = \frac{\partial}{\partial \theta_\alpha}- \frac{1}{4}\theta_\beta(C\gamma^\mu)^{\beta \alpha}\frac{\partial}{\partial x^\mu}.$$
We chose the distribution spanned by $P_\mu = \frac{\partial }{\partial x^\mu}$ as the complementary distribution.   It is easy to see that these satisfy the super-translation algebra
\begin{align}\label{Eqn:SUSYALg}
&[D^\alpha, D^\beta] = - \frac{1}{2} (C\gamma^\mu)^{ \alpha \beta}P_\mu,
& [P_\mu, D_\alpha] =0, && [P_\mu, P_\nu] =0.
\end{align}
Any vector field on $\textnormal{SMink}^{4|4}$ decomposes as $X = X^\mu(x, \theta)P_\mu + X_\alpha(x, \theta)D^\alpha$. We then define an odd involution
$\rho :\Vect(\textnormal{SMink}^{4|4}) \rightarrow \Vect(\textnormal{SMink}^{4|4}) $ as (assuming $X$ is homogeneous)
$$\rho(X) = (-1)^{\widetilde{X}}\left( X^\mu \delta_{\mu \alpha}D^\alpha - X_\alpha \delta^{\alpha \mu}P_\mu \right).$$
\begin{definition}\label{Def:SUSYOddCon}
The \emph{SUSY odd connection} on $\textnormal{SMink}^{4|4}$ is defined as
$$\nabla_X Y := (-1)^{\widetilde{X}}\left( X^\mu \delta_{\mu \alpha}D^\alpha  Y^\nu - X_\alpha \delta^{\alpha \mu}P_\mu Y^\nu \right)P_\nu + (-1)^{\widetilde{X}}\left( X^\mu \delta_{\mu \alpha}D^\alpha Y_\beta - X_\alpha \delta^{\alpha \mu}P_\mu Y_\beta \right)D^\beta,$$
for all homogeneous $X \in \Vect(\textnormal{SMink}^{4|4})$ and all  $Y \in \Vect(\textnormal{SMink}^{4|4})$.
\end{definition}
\begin{remark}
The SUSY odd connection is similar to but not identical to the canonical odd connection on $\R^{4|4}$ as given in Example \ref{Exp:CanOddCon}. The reader should also compare this with Example \ref{Exp:1dSUSY}, which is, of course,  the corresponding $1|1$-dimensional case.
\end{remark}
\begin{proposition}\label{Prop:SMinkFlat}
The SUSY odd connection on  $\textnormal{SMink}^{4|4}$ is flat (see Definition \ref{Def:Flat}).
\end{proposition}
\begin{proof}
Let us, for brevity, set $e_a = (P_\mu , D^\alpha)$ and so an arbitrary vector field we write as $Z = Z^a e_a$. First observe, directly from the definition of the SUSY odd connection (Definition \ref{Def:SUSYOddCon}) and the fact that $\rho$ is an involution, that
$$\nabla_{\rho([\rho(X) , \rho(Y)])}Z = \big(\rho(X)(\rho(Y)Z^a)\big)e_a - (-1)^{(\widetilde{X}+1)(\widetilde{Y}+1) } \big(\rho(Y)(\rho(X)Z^a)\big)e_a.$$
Then from the definition of the  curvature (Definition \ref{Def:RieCurv}) we see that
\begin{align*}
R(X,Y)Z &= \big(\rho(X)(\rho(Y)Z^a)\big)e_a - (-1)^{(\widetilde{X}+1)(\widetilde{Y}+1) } \big(\rho(Y)(\rho(X)Z^a)\big)e_a\\
& -\nabla_{\rho([\rho(X) , \rho(Y)])}Z = \big(\rho(X)(\rho(Y)Z^a)\big)e_a,
\end{align*}
and so $R(X,Y)Z =0$ for arbitrary vector fields $X, Y$ and $Z \in \Vect(\textnormal{SMink}^{4|4})$. Thus, the SUSY odd connection is flat.
\end{proof}
\begin{proposition}\label{Prop:SMinkTorsion}
The SUSY odd connection on  $\textnormal{SMink}^{4|4}$ has non-vanishing torsion (see Definition \ref{Def:Torsion}).
\end{proposition}
\begin{proof}
As the torsion is a tensor it is sufficient to check its action on  pairs of $P_\mu$ and $D^\alpha$. Direct calculation using the super-translation algebra \eqref{Eqn:SUSYALg} gives
\begin{align*}
& T(P_\mu , P_\nu) =  - \frac{1}{2} (C\gamma^\delta)^{\alpha \beta} \delta_{\beta \nu} \delta_{ \alpha \mu} \delta_{\delta \gamma} D^\gamma,\\
& T(P_\mu , D^\alpha) = - \frac{1}{4} \delta_{\mu \beta} (C\gamma^\nu)^{\beta \alpha}P_\nu \\
& T(D^\alpha, D^\beta) =0.
\end{align*}
Clearly, not all of these vanish and we conclude that the torsion is non-zero.
\end{proof}
\begin{remark}
We observe that, and this is not at all surprising, that the non-zero components of the torsion are essentially $ -\half(C \gamma^\mu)^{\alpha \beta}$, which is just the non-vanishing structure constant of the super-translation algebra.
\end{remark}
\begin{definition}
The \emph{odd divergence operator} on  $\textnormal{SMink}^{4|4}$ is defined as
$$\textnormal{Div}\, X = D^\alpha(X^\mu) \delta_{\mu \alpha} - (-1)^{\widetilde{X}} \,\delta^{\alpha \mu}P_\mu(X_\alpha)\,,$$
for any (homogeneous) vector field $X =  X^\mu P_\mu + X_\alpha D^\alpha$.
\end{definition}

It is clear that, as one requires an equal number of even and odd coordinates, that the constructions given above do not generalise directly to $d = 4$, $\mathcal{N} \geq 2$ extended super-Minkowski space-times. Requiring an equal number of even and odd coordinates places restrictions on the dimension of the underlying Minkowski space-time. For example, $d=1$, $\mathcal{N}=1$ super-Minkowski space-time has dimension $1|1$ and so, as we have seen (Example \ref{Exp:1dSUSY} and Example \ref{Exp:OddDiv1dSUSY}), the construction of the  odd SUSY connection generalises to this case. Importantly, the dimension of the real irreducible spin representation in one dimension is one, see Freed \cite[page 48]{Freed:1999} for details of the dimensions of real spin representations.   The case for $d=2$ is slightly more complicated as we have two one-dimensional irreducible spin representations, and  so choosing $\mathcal{N} = (1,1), (2,0)$ and $(0,2)$ will allow the construction of an odd SUSY connection.  In comparison, $d = 3$ has a two-dimensional real irreducible spin representation. Thus, the $d=3$ case does not permit the direct construction of an odd SUSY connection.  Assuming that $d \leq 11$, we have exhausted the list of possible dimensions and number of supersymmetries that one can directly construct an odd connection. Simply put: in other dimensions, it is impossible to have an equal number of  even coordinates and odd spinor coordinates.

\subsection{Odd Weitzenb\"{o}ck Connections}\label{SubSec:OddWeitCon}
The SUSY odd connection (see Definition \ref{Def:SUSYOddCon}) is built from just an odd involution on the module of vector fields. The same is true of the canonical odd connection on $\R^{n|n}$ (see Example \ref{Exp:CanOddCon}). Moreover, we see that the curvature of these connections is zero (see Proposition \ref{Prop:SMinkFlat}), while the torsion is not zero (see Proposition \ref{Prop:SMinkTorsion}). This is very reminiscent of  the notion of a  Weitzenb\"{o}ck connection, as used in teleparallel gravity and related theories where gravity is ``all torsion and no curvature".  We also remark that  Weitzenb\"{o}ck connections make an appearance in Double Field Theory   (see for example \cite{Penas:2019} and references therein). These considerations lead to the following notion.
\begin{definition}\label{Def:OddWeitzCon}
Let $M$ be a $n|n$-dimensional parallelisable supermanifold and let $\{Z_\alpha \}$ be a chosen  parallelisation. Furthermore, let $\rho : \Vect(M) \rightarrow \Vect(M)$ be an odd involution. The \emph{odd Weitzenb\"{o}ck connection} on $M$ generated by $\rho$ and $\{Z_\alpha \}$ is the odd connection defined as
$$\nabla_X(Y^\alpha Z_\alpha):= \rho(X)(Y^\alpha)\, Z_\alpha.$$
\end{definition}
\begin{proposition}\label{Prop:WeitzIndPara}
The odd  Weitzenb\"{o}ck connection on a $n|n$-dimensional parallelisable supermanifold $M$ generated by $\rho$ is independent of the chosen  parallelisation.
\end{proposition}
\begin{proof}
As a choice of parallelisation corresponds to an isomorphism $\Vect(M)\stackrel{\sim}{ \rightarrow} C^\infty(M) \otimes \vec{\R}^{n|n}$, changes of the parallelisation  correspond to grading preserving automorphisms of the super vector space  $\vec{\R}^{n|n}$. Thus, if we have two parallelisations $\{Z_\alpha\}$ and  $\{Z_{\beta'}\}$, then there is an invertible matrix $A$ with real entries, such that   $Z_{\beta'} = A_{\beta'}^{\,\,\, \alpha} Z_\alpha$. This, in turn, implies that the components of the vector fields transform via the inverse matrix, i.e.,  $Y^{\beta'} = Y^\alpha A_{\alpha}^{\,\, \beta'}$. Then we observe that
\begin{align*}
\nabla_X Y & = \rho(X)(Y^{\beta'}Z_{\beta'}) = \rho(X)(Y^{\beta'})Z_{\beta'}\\
&= \rho(X)( Y^\alpha A_{\alpha}^{\,\, \beta'})A_{\beta'}^{\,\,\, \gamma} Z_{\gamma} = \rho(X)( Y^\alpha) A_{\alpha}^{\,\, \beta'}A_{\beta'}^{\,\,\, \gamma} Z_{\gamma}\\
& = \rho(X)(Y^{\alpha}Z_{\alpha}),
\end{align*}
where we have used the fact that the components of $A^{-1}$ are constants.
\end{proof}
The above proposition tells us that an odd  Weitzenb\"{o}ck connection is completely defined by a choice of odd involution and is does not depend on the choice  parallelisation.  For the remaining part of this subsection, we will for brevity denote a supermanifold equipped with an odd Weitzenb\"{o}ck connection as a pair $(M, \nabla)$.
\begin{proposition}\label{Prop:ExistWeitz}
 Let $M$ be a $n|n$-dimensional parallelisable supermanifold $M$. Then the set of odd  Weitzenb\"{o}ck connections on $M$ is non-empty.
\end{proposition}
\begin{proof}
This is clear as the set of odd involutions on a $n|n$-dimensional parallelisable supermanifold is non-empty and Proposition \ref{Prop:WeitzIndPara} tells us that this all that is needed to define an odd  Weitzenb\"{o}ck connection.\\
\end{proof}
In general, the torsion on an odd   Weitzenb\"{o}ck connection will be non-zero. However, just as in the classical setting, the curvature  is zero.
\begin{proposition}
An odd   Weitzenb\"{o}ck connection on a $n|n$-dimensional parallelisable supermanifold $M$ is flat (see Definition \ref{Def:Flat}).
\end{proposition}
\begin{proof}
The proof is identical to the proof of Proposition \ref{Prop:SMinkFlat} upon minor notational changes.\\
\end{proof}
In complete parallel with the classical setting of smooth manifolds, the Christoffel symbols of an odd   Weitzenb\"{o}ck connection are given in terms of local vierbein fields and their derivative.  To recall the notion, let us start with a $n|m$-dimensional parallelisable supermanifold $M$ and let $\{Z_\alpha\}$ be a choice of parallelisation. Locally we define the vierbeins $E_\alpha^{\,\, \,a}(x)$, where $\widetilde{E_\alpha^{\,\, \,a}} = \widetilde{a} + \widetilde{\alpha}$, via
$$Z_\alpha =  (-1)^{\widetilde{a} \, \widetilde{\alpha}} \,E_\alpha^{\,\, \,a}(x)\frac{\partial}{\partial x^a}\,. $$
The sign factor is included for convenience. Dual to the global frame for the vector fields is a global basis for the one-forms, which we denote as $\{\omega^\alpha\}$, which consists of $n$ even one-forms and $m$ odd one-forms.  The co-vierbeins are similarly defined locally as
$$\omega^\alpha = (-1)^{\widetilde{a} \, \widetilde{\alpha}} \, \delta x^a \, E_a^{\,\,\, \alpha}(x),$$
where we have chosen the convention that $\widetilde{\delta x^a} = \widetilde{a}$, and so $\widetilde{E_a^{\,\, \,\alpha}} = \widetilde{a} + \widetilde{\alpha}$. We have the standard orthonormality conditions which are directly deduced from $\omega^\alpha(Z_\beta) = \delta^\alpha_{\,\,\, \beta}$ and $\delta x^a(\partial_b) = \delta^a_{\,\,\, b} $,
\begin{align}\label{Eqn:VierOrtNorm}
E_a^{\,\,\, \alpha}E_\beta^{\,\, \,a} = \delta^\alpha_{\,\,\, \beta}, && E_\alpha^{\,\, \,a}E_b^{\,\,\, \alpha} = \delta^a_{\,\,\, b}.
\end{align}
\begin{proposition}
The Christoffel symbols of an odd  Weitzenb\"{o}ck connection are given by
$$\Gamma_{ba}^{\,\,\,\, c} = (-1)^{\widetilde{b}(\widetilde{c} +1) + \widetilde{d}(\widetilde{c} + \widetilde{\alpha})}\, \rho_a^{\,\, d} \, E_\alpha^{\,\,\, c}\left( \frac{\partial E_b^{\,\, \alpha}}{\partial x^d} \right ) .$$
\end{proposition}
\begin{proof}
the proof follows in more-or-less the same way as the classical case. Directly from the Leibniz rule we see that
$$\nabla_X(Y^\alpha Z_\alpha) = \rho(X)Y^\alpha \, Z_\alpha + (-1)^{(\widetilde{X}+1)(\widetilde{Y} + \widetilde{\alpha})} \, Y^\alpha \, \nabla_X Z_\alpha,$$
and hence the final term must vanish, i.e., $\nabla_X Z_\alpha =0$. It is this result that will determine the Christoffel symbols. Locally  using the vierbeins and the fact that $X \in \Vect(M)$ is arbitrary we see that the previous result amounts to
$$ (-1)^{\widetilde{c} \, \widetilde{\alpha}} \, \rho_a^{\,\,\, b} \frac{\partial E_\alpha^{\,\,\, c}}{\partial x^b} = - (-1)^{(\widetilde{a}+1)(\widetilde{\alpha} + \widetilde{b}) + \widetilde{b} \, \widetilde{\alpha}} \, E_\alpha^{\,\,\, b}\Gamma_{ba}^{\,\,\,\, c}\,.$$
Now we multiply by $E_d^{\,\,\, \alpha}$ from the right, pull it through the Christoffel symbol and use the orthonormality of the vierbeins and co-vierbeins  to get
\begin{equation}\label{Eqn:WeitChrsA}
\Gamma_{ba}^{\,\,\,\, c} = - (-1)^{\widetilde{d}(\widetilde{c} +1)}\, \rho_a^{\,\, d}\, \left( \frac{\partial E_\alpha^{\,\,\, c}}{\partial x^d} \right) E_b^{\,\,\, \alpha}.
\end{equation}
Using the orthonormality, it is clear that
$$\left(  \frac{\partial E_\alpha^{\,\,\, c}}{\partial x^d} \right) E_b^{\,\,\, \alpha} = - (-1)^{\widetilde{d}(\widetilde{c} + \widetilde{\alpha})}\, E_\alpha^{\,\,\, c} \left( \frac{\partial E_b^{\,\,\, \alpha}}{\partial x^d} \right).$$
Substituting this into \eqref{Eqn:WeitChrsA} produces the desired result
$$\Gamma_{ba}^{\,\,\,\, c} = (-1)^{\widetilde{b}(\widetilde{c} +1) + \widetilde{d}(\widetilde{c} + \widetilde{\alpha})}\, \rho_a^{\,\, d}   \, E_\alpha^{\,\,\, c}\left( \frac{\partial E_b^{\,\, \alpha}}{\partial x^d} \right ) .$$
\end{proof}

\begin{definition}\label{Def:OddMetric}
Let $M$ be an $n|n$-dimensional parallelisable supermanifold and let  $\{ X_i \,; Y_j\}$, where $i,j = 1,2, \cdots n$ be a chosen  parallelisation.  The \emph{induced odd Riemannian metric} on $M$ is defined as $g(X_i, X_j) = g(Y_i, Y_j) =0$ and $g(X_i, Y_j) = \delta_{ij}$.
\end{definition}
\begin{proposition}\label{Prop:OddWeitComOddMet}
An odd  Weitzenb\"{o}ck connection on a supermanifold $M$ (see Definition \ref{Def:OddWeitzCon}) is compatible with the induced odd metric (see Definition \ref{Def:OddMetric} and  Definition \ref{Def:OddConCompG}).
\end{proposition}
\begin{proof}
 This follows from the  corresponding statement for Weitzenb\"{o}ck connections and the fact that an odd connection can canonically be obtained from an affine. Alternatively, this follows from direct calculation.  We write the chosen parallelisation  as $\{Z_\alpha\}$ and write for $Y_1 = Y_1^\alpha Z_\alpha$ and $Y_2 = Y_2^\beta Z_\beta$ for two homogeneous but otherwise arbitrary vector fields. Then via direct calculation
\begin{align*}
\rho(X) \big (g(Y_1, Y_2) \big )  &=  \rho(X)\big( g(Y_1^\alpha Z_\alpha, Y_2^\beta Z_\beta) \big) =  (-1)^{\widetilde{\alpha}(\widetilde{Y}_2 + \widetilde{\beta})} \, \rho(X)\big( Y_1^\alpha Y_2^\beta \,  g( Z_\alpha ,Z_\beta) \big)\\
&=  (-1)^{\widetilde{\alpha}(\widetilde{Y}_2 + \widetilde{\beta})} \,  (\rho(X)Y_1^\alpha) Y_2^\beta \,  g( Z_\alpha ,Z_\beta)  + (-1)^{\widetilde{\alpha}(\widetilde{Y}_2 + \widetilde{\beta}) + (\widetilde{X}+1)(\widetilde{Y}_1 + \widetilde{\alpha})} \,  Y_1^\alpha(\rho(X) Y_2^\beta) \,  g( Z_\alpha ,Z_\beta) \\
& = g(\nabla_X Y_1, Y_2) + (-1)^{(\widetilde{X}+1)\widetilde{Y}_1} \, g(Y_1, \nabla_X Y_2),
\end{align*}
where we have used the fact that the odd Riemannian metric is constant in the chosen basis. Comparing this with \eqref{Eqn:OddConCompG} we see establish the proposition.\\
\end{proof}

\section{Closing Remarks}
Starting from  quite general  odd quasi-connections, we constructed the notion of an odd connection on a supermanifold. Importantly, for odd connections, the torsion and  curvature are tensors, a property we expect any reasonable generalised notion of an affine connection to have. In particular, the tensorial property of the torsion and curvature guarantee that they can locally be written in terms of their components with respect to a chosen basis, say the coordinate basis.  We explored the relationship  between odd connections and affine connections on a supermanifold  by realising that, up to a tensorial term, any given odd connection can be induced by an arbitrary affine connection. This was then used to tackle the question of the existence of odd connections. While affine connections always exist, there are severe restrictions on the supermanifolds that admit odd connections. The existence of an odd involution of the module of vector fields on the supermanifold is essential in the definition of an odd connection.  For example, $n|n$-dimensional Lie supergroups and, more generally, $n|n$-dimensional parallisable supermanifolds always admit odd connections.  The prototypical Lie supergroups here is $\textnormal{GL}(m|m)$, which is of dimension $2m^2|2m^2$.  Moreover, we have shown that $\mathcal{N} =1$ super-Minkowski spacetime comes equipped with a natural odd connection. This gives a slightly different perspective on the geometric nature of  supersymmetry.  The example of super-Minkowski spacetime leads to odd Weitzenb\"{o}ck connections on $n|n$-dimensional parallisable supermanifolds and it was shown that these odd connections depend only on the choice of odd involution. \par
 Applications in physics are probably limited due to the need for an equal number even and odd dimensions alongside the supermanifolds being parallisable. However, one can construct novel field theories using odd connections. For example, if $(M,g)$  is a parallisable even Riemannian supermanifold of dimension $n|n$, then one can construct an action for an odd scalar field $\Psi(x)$, which we write locally as (being slack with an overall sign)
$$S[\Psi]= \pm \int_M D[x]  \,\sqrt{|g|}  ~ (-1)^{\widetilde{b}}\, \half  g^{ab}(x)\nabla_b \Psi(x) \nabla_a \Psi(x),$$
where $\nabla$ is the odd connection canonically generated by the Levi-Civita connection (for example) and a choice of odd involution. Note that using a odd connection is essential here as otherwise the action would be identically zero: we would be contracting something symmetric in its indices with something antisymmetric in its indices.  Other theories can similarly defined such as even scalar field theories on an odd Riemannian supermanifold or a symplectic/Poisson  supermanifold. While such theories deserve further attention, to do so is outside the scope of this paper.

\section*{Acknowledgements}
The authors thank the anonymous referee for their valuable comments that have served to improve the overall presentation of this work.  J. Grabowski acknowledges research founded by the  Polish National Science Centre grant HARMONIA under the contract number 2016/22/M/ST1/00542.

\newpage

\appendix

\section{}

\subsection{Proof of the Algebraic Bianchi Identity}\label{App:ProofAlgBian}
\begin{proof}[Proof of Theorem \ref{Thm:Bianchi1}] The proof is via direct computation following the proof of  the standard algebraic Bianchi identity. Let $X,Y$ and $Z \in \Vect(M)$ be homogeneous. Then directly from the definition of the  curvature (Definition \ref{Def:RieCurv}) we have
\begin{align*}
&(-1)^{\widetilde{X}(\widetilde{Z} +1)}\,R(X,Y)Z + (-1)^{\widetilde{Y}(\widetilde{X} +1)}\,R(Y,Z)X  + (-1)^{\widetilde{Z}(\widetilde{Y} +1)}\,R(Z,X)Y \\
& =(-1)^{\widetilde{X}(\widetilde{Z} +1)}\, \nabla_X \nabla_Y Z - (-1)^{\widetilde{X}(\widetilde{Z} +1) + (\widetilde{X}+1)(\widetilde{Y}+1)}\,\nabla_X \nabla_Y Z  - (-1)^{\widetilde{X}(\widetilde{Z} +1)}\, \nabla_{\rho[\rho(X), \rho(Y)]}Z\\
& + (-1)^{\widetilde{Y}(\widetilde{X} +1)}\, \nabla_Y \nabla_Z X - (-1)^{\widetilde{Y}(\widetilde{X} +1) + (\widetilde{Y}+1)(\widetilde{Z}+1)}\,\nabla_Z \nabla_Y X  - (-1)^{\widetilde{Y}(\widetilde{X} +1)}\, \nabla_{\rho[\rho(Y), \rho(Z)]}X\\
& + (-1)^{\widetilde{Z}(\widetilde{Y} +1)}\, \nabla_Z \nabla_X Y - (-1)^{\widetilde{Z}(\widetilde{Y} +1) + (\widetilde{Z}+1)(\widetilde{X}+1)}\,\nabla_X \nabla_Z Y  - (-1)^{\widetilde{Z}(\widetilde{Y} +1)}\, \nabla_{\rho[\rho(Z), \rho(X)]}Y \\
\\
&= (-1)^{\widetilde{X}(\widetilde{Z} +1)}\,\nabla_X \big(\nabla_YZ +(-1)^{\widetilde{Z}\widetilde{Y}}\, \nabla_Z Y \big) - (-1)^{\widetilde{X}(\widetilde{Z}+1)}\,\nabla_{\rho[\rho(X), \rho(Y)]}Z\\
 &+ (-1)^{\widetilde{Y}(\widetilde{X} +1)}\,\nabla_Y \big(\nabla_Z X +(-1)^{\widetilde{X}\widetilde{Z}}\, \nabla_X Z \big) - (-1)^{\widetilde{Y}(\widetilde{X}+1)}\,\nabla_{\rho[\rho(Y), \rho(Z)]}X\\
 &+ (-1)^{\widetilde{Z}(\widetilde{Y} +1)}\,\nabla_Z \big(\nabla_X Y +(-1)^{\widetilde{Y}\widetilde{X}}\, \nabla_Y X \big) - (-1)^{\widetilde{Z}(\widetilde{Y}+1)}\,\nabla_{\rho[\rho(X), \rho(Y)]}Z,\\
\\
&\textnormal{now using the definition of the torsion (Definition \ref{Def:Torsion}) we rewtite this as}\\
\\
&= (-1)^{\widetilde{X}(\widetilde{Z} +1)}\,\nabla_X \big(T(Y,Z) - (-1)^{\widetilde{Y}} \rho[\rho(Y), \rho(Z)] \big) - (-1)^{\widetilde{X}(\widetilde{Z}+1)}\,\nabla_{\rho[\rho(X), \rho(Y)]}Z\\
&+ (-1)^{\widetilde{Y}(\widetilde{X} +1)}\,\nabla_Y \big((T(Z,X) - (-1)^{\widetilde{Z}} \rho[\rho(Z), \rho(X)] \big) - (-1)^{\widetilde{Y}(\widetilde{X}+1)}\,\nabla_{\rho[\rho(Y), \rho(Z)]}X\\
&+ (-1)^{\widetilde{Z}(\widetilde{Y} +1)}\,\nabla_Z \big((T(X,Y) - (-1)^{\widetilde{X}} \rho[\rho(X), \rho(Y)] \big) - (-1)^{\widetilde{Z}(\widetilde{Y}+1)}\,\nabla_{\rho[\rho(Z), \rho(X)]}Y,\\
\\
&= (-1)^{\widetilde{X}(\widetilde{Z} +1)}\, \nabla_X\big( T(Y,Z)\big)+ (-1)^{\widetilde{Y}(\widetilde{X} +1)}\, \nabla_Y\big( T(Z,X)\big)+ (-1)^{\widetilde{Z}(\widetilde{Y} +1)}\, \nabla_Z\big( T(X,Y)\big)\\
&- (-1)^{\widetilde{X}(\widetilde{Z} +1) + \widetilde{Y}}\big( \nabla_X \rho[\rho(Y), \rho(Z) ]  + (-1)^{\widetilde{X}(\widetilde{Y} +\widetilde{Z} +1)}  \, \nabla_{\rho[\rho(Y), \rho(Z)]}X\big)\\
&- (-1)^{\widetilde{Y}(\widetilde{X} +1) + \widetilde{Z}}\big( \nabla_Y \rho[\rho(Z), \rho(X) ]  + (-1)^{\widetilde{Y}(\widetilde{Z} +\widetilde{X} +1)}  \, \nabla_{\rho[\rho(Z), \rho(X)]}Y\big)\\
&- (-1)^{\widetilde{Z}(\widetilde{Y} +1) + \widetilde{X}}\big( \nabla_Z \rho[\rho(X), \rho(Y) ]  + (-1)^{\widetilde{Z}(\widetilde{X} +\widetilde{Y} +1)}  \, \nabla_{\rho[\rho(X), \rho(Y)]}Z\big),\\
\\
&\textnormal{using the definition of the torsion (Definition \ref{Def:Torsion}) again and a little rearanging}\\
\\
&= (-1)^{\widetilde{X}(\widetilde{Z} +1)}\, \nabla_X\big( T(Y,Z)\big)+ (-1)^{\widetilde{Y}(\widetilde{X} +1)}\, \nabla_Y\big( T(Z,X)\big)+ (-1)^{\widetilde{Z}(\widetilde{Y} +1)}\, \nabla_Z\big( T(X,Y)\big)\\
&-(-1)^{\widetilde{X}(\widetilde{Z} +1) + \widetilde{Y}}\, T\big( X, \rho[\rho(Y), \rho(Z)]  \big)-(-1)^{\widetilde{Y}(\widetilde{X} +1) + \widetilde{Z}}\, T\big( Y, \rho[\rho(Z), \rho(X)]  \big)\\
&-(-1)^{\widetilde{Z}(\widetilde{Y} +1) + \widetilde{X}}\, T\big( Z, \rho[\rho(X), \rho(Y)]  \big)\\
&+ \rho \big( [\rho(X),[\rho(Y), \rho(Z)]] - [[\rho(X), \rho(Y)] , \rho(Z)]-(-1)^{(\widetilde{X}+1)(\widetilde{Y}+1)} \, [\rho(Y),[\rho(X), \rho(Z)]]\big).
\end{align*}
The final term vanishes due to the Jacobi identity (here written in Loday--Leibniz form).\\
\end{proof}
\vfill
\end{document}